\documentclass[arxiv]{imsart}

\RequirePackage{amsthm,amsmath,amsfonts,amssymb}
\RequirePackage{natbib}
\RequirePackage[colorlinks,citecolor=blue,urlcolor=blue]{hyperref}
\usepackage{bm}
\usepackage{multirow}
\usepackage{float}
\usepackage{graphicx}
\usepackage{mathtools}
\usepackage{color}
\usepackage[linesnumbered,ruled,vlined]{algorithm2e}
\usepackage{booktabs}
\usepackage{geometry} 
\newcommand{\R}{\mathbb{R}}
\DeclareMathOperator*{\argmin}{arg\,min}

\startlocaldefs
\theoremstyle{remark}
\newtheorem{thm}{Theorem}

\theoremstyle{remark}
\newtheorem*{remark*}{Remark}

\newcommand{\ignore}[1]{}

\endlocaldefs

\begin{document}

\pagestyle{plain} 

\begin{frontmatter}

\title{A Robust Local Fr\'echet Regression Using Unbalanced Neural Optimal Transport with Applications to Dynamic Single-cell Genomics Data}

\begin{aug}
    \author{\fnms{Binghao} \snm{Yan}\ead[label=e1]{Binghao.Yan@Pennmedicine.upenn.edu}},
    \author{\fnms{Hongzhe} \snm{Li}\thanks{Corresponding author}\ead[label=e2]{hongzhe@upenn.edu}}
    \address{Department of Biostatistics, Epidemiology, and Informatics, University of Pennsylvania, \printead{e1,e2}}
\end{aug}

\begin{abstract}
Single-cell RNA sequencing (scRNA-seq) technologies have enabled the profiling of gene expression for a collection of cells across time during a dynamic biological process. Given that each time point provides only a static snapshot, modeling and understanding the underlying cellular dynamics remains a central yet challenging task in modern genomics. To associate biological time with single cell distributions, we develop a robust local Fr\'echet regression for interpolating the high-dimensional cellular distribution at any given time point using data observed over a finite time points. To allow for robustness in cell distributions, we propose to apply the unbalanced optimal transport-based Wasserstein distance in our local Fr\'echet regression analysis. We develop a computationally efficient algorithm to generate the cell distribution for a given time point using generative neural networks. The resulting single cell generated models and the corresponding transport plans can be use to interpolate the single cells at any unobserved time point and to track the cell trajectory during the cell differentiation process. We demonstrate the methods using three single cell differentiation data sets, including differentiation of human embryonic stem cells into embryoids, mouse hematopoietic and progenitor cell differentiation, and reprogramming of mouse embryonic fibroblasts.  We show that the proposed methods lead to better single cell interpolations, reveal different cell differential trajectories, and identify early genes that regulate these cell trajectories. 

\end{abstract}


\end{frontmatter}

\section{Introduction}

In many areas of modern data science, we increasingly encounter distributional data paired with covariates of interest—such as time, age, or experimental conditions. An important and rapidly growing example arises in single-cell RNA sequencing (scRNA-seq) studies to understand the dynamic cell differentiation process, where researchers profile single-cell-level transcriptomic data across multiple time points \citep{moon2019embryoid, weinreb2020statefate, schiebinger2019optimal}. At each time point, a population of cells are profiled, resulting in a high-dimensional distribution of cells characterized by their gene expressions. These cellular distributions reflect dynamic biological processes such as development, aging, or reprogramming. In this setting, each collection of cells can be viewed as a sample from a probability distribution conditioned on a covariate (e.g., time). Such data contains rich information about how biological systems evolve. However, most existing methods focus on gene-level or cell-level analyses, which do not capture the full geometry and structure of the underlying cellular distributions. Statistical models to extract insights from these distributional datasets are greatly needed. 

In this paper, we focus on scRNA-seq datasets measured over time during a biological process, where time can be treated as a covariate.  Understanding how cell gene expression distributions evolve over time is critical for uncovering lineage relationships and disease mechanisms. Several challenges arise in this setting. First, single-cell technologies are inherently destructive—individual cells cannot be tracked over time—making it difficult to directly analyze the cell dynamics. Second, due to experimental cost, observations are often available only at a few time points. Third, noise and technical variability can make the empirical distribution at each time point an imperfect reflection of the true biological state. Lastly, the high dimensionality of the data poses significant computational challenges for statistical modeling and inference.

{Optimal transport} (OT), seeking the most cost-efficient way to rearrange one distribution into another by minimizing the total ``movement” cost between their mass points \citep{kantorovich1942, villani2008optimal}, has emerged as a powerful tool for studying cell dynamics, with several successful applications in this area. Waddington-OT \citep{schiebinger2019optimal} leverages optimal transport to learn temporal couplings between adjacent time snapshots, effectively uncovering the dynamics during the cellular reprogramming process. However, a notable limitation is its reliance on only two adjacent distributions, which results in a loss of long-range information. JKOnet \citep{bunne2022proximal} integrates Jordan–Kinderlehrer–Otto (JKO) flow with optimal transport to predict system dynamics but lacks the ability to interpolate between the two time snapshots. TrajectoryNet \citep{tong2020trajectorynet} explores unbalanced dynamic transport and employs neural ODEs to recover trajectories. MioFlow \citep{huguet2022manifold} improves on TrajectoryNet by incorporating diffusion and manifold learning into the framework. Despite their strengths, we observe that in noisy real-world datasets, the learned trajectories often exhibit bias due to the noise, particularly at later time points.

In this work, we propose a new approach based on {robust local Fr\'echet regression}  for analyzing cellular dynamics based on scRNA-seq data measured at multiple time points. Rather than directly inferring individual cell trajectories, our goal is to model and interpolate the cell gene expression distributions  over time. Given a time value $t$, we estimate the corresponding distribution using observed samples from nearby time points. To gain robustness of the estimate,  we use the Wasserstein-distance defined from the unbalanced optimal transport (UOT) \citep{chizat2018unbalanced, choi2024generative} between two probability measures instead of the standard OT. By penalizing mass creation/destruction rather than forcing exact transport, unbalanced OT is less sensitive to small-scale noise or outlier spikes in either distribution. This yields more stable mappings and distances in analysis of single cell data.

Since the output from the local Frechet regression is a probability measure in high dimension, we solve the regression  objective using a neural network-based fixed-point algorithm, where 
we extend the classical neural fixed-point algorithm \citep{korotin2022wasserstein} to the unbalanced setting, improving robustness to noise and algorithm convergence in real datasets. At the core of our approach is a reparametrization of the estimand via a generative model, which converts an intractable optimization over probability measures into a more manageable optimization over the model’s parameters. 

Our method enjoys several advantages: (1) it is {nonparametric}, relying only on kernel-weighted averages of distributions and making no assumptions about the underlying dynamics; (2) it incorporates ideas from unbalanced optimal transport, making it well-suited for noisy, imperfect single-cell measurements; and (3) it {scales to high-dimensional data} through neural network parameterizations.

We demonstrate the effectiveness of our method using  both synthetic and real datasets. In addition to more accurate interpolation of cell gene expressions, we show that the estimated distributions can be used to reconstruct individual cell trajectories through composing the unbalanced OT maps between interpolated distributions. This allows a better recovery of the developmental trajectories and the identification of lineage fates and marker genes.

The rest of the paper is organized as follows. In Section \ref{ROLFR}, we provide a brief comparison between optimal transport and unbalanced optimal transport based Wasserstein distances before we introduce the robust local Fr\' echet regression. We then provide a fixed point algorithm  to solve the local Fr\' echet regression estimation in Section \ref{Fixed-point}. We provide more detailed algorithm to solve this fixed-point algorithm using neural network parameterizations in Section \ref{NeuralNet}. In Section \ref{Realdata}, 
we present detailed analysis of three single cell data sets during different cell differentiation dynamic processes.  Finally, in Section \ref{Discussion}, we  provide a brief discussion of the methods and possible extensions.

\ignore{
Our main contributions are as follows:
\begin{itemize}
    \item We introduce a flexible framework for interpolating cell population distributions over time using local Fr\'echet regression.
    \item We extend the classical neural fixed-point algorithm to the unbalanced setting, improving robustness to noise and convergence in real datasets.
    \item We validate our method on three real scRNA-seq temporal datasets, demonstrating good interpolation quality and enabling inference of cell fate decisions and associated genes.
\end{itemize}
}

\section{A Robust Local Fr\' echet Regression}
\label{ROLFR}
We consider the setting of  scRNA-seq studies with gene expressions measured at $N$ time points during a dynamic biological process, where at each time point $t_i$, we have $n_i$ cells with $p$ genes measured. These $p$ gene expressions characterize the cellular states. In practice, the low-dimensional representations using principal components  (PCs) or tSNEs are used to quantify the cell state.   In this paper, we mainly use $d$ PCs to represent the cell state. At each time point $t_i$, let $\nu_i$
be the empirical distribution of these $d$ PCs of all the cells measured. We are interested in learning how such cellular distribution $\nu$ changes over time based on the observed snapshots of the cells measured, represented by their empirical distributions $\nu_i, i=1, \cdots, N$.

We follow the standard  notation in optimal transport. Let $\mathcal{X}, \mathcal{Y}$ be compact and complete subspaces of $\mathbb{R}^d$. Let $\mathcal{P}_2(\mathcal{X})$ denote the set of all probability measures on $\mathcal{X}$ with finite second moments, and let $\mathcal{P}_{2,ac}(\mathcal{X})$ denote the subset of probability measures that are absolutely continuous with respect to the Lebesgue measure.

\subsection{Balanced and unbalanced optimal transport}

We begin with an introduction to the fundamental concepts of optimal transport. The optimal transport  between two probability measures aims at finding the least cost way to transport the source measure to the target measure. Given two probability measures   $\mu \in \mathcal{P}_{2,ac}(\mathcal{X})$ and  $\nu \in \mathcal{P}_{2,ac}(\mathcal{Y})$, the Kantorovich formulation \citep{kantorovich1942} for the optimal transport problem and its corresponding Wasserstein distance (squared) is
\begin{equation}
    W_2^2(\mu, \nu) \coloneqq \inf_{\pi \in \Pi(\mu, \nu)} \int_{\mathcal{X} \times \mathcal{Y}} \frac{1}{2} \|x - y\|^2 d\pi(x, y),
    \label{eq:kantorovichOT}
\end{equation}
where the infimum is taken over all transport plans $\pi$, i.e., the coupling of $\mu$ and $\nu$. Under mild conditions on $(\mu,\mathcal{X})$, $(\nu,\mathcal{Y})$, the minimizer {optimal transport plan} $\pi^*$ always exists \citep{villani2008optimal}. 

The Kantorovich formulation~\eqref{eq:kantorovichOT} admits the following {semi-dual} form. The functions $u$ and $v$ are called {potentials}. There exist optimal $u^*, v^*$ satisfying $u^* = (v^*)^c$, where $f^c(y) \coloneqq \min_{x \in \mathbb{R}^D} \left[ \frac{1}{2} \|x - y\|^2 - f(x) \right]$ is the $c$-transform of $f$. We rewrite \eqref{eq:kantorovichOT}  as
\begin{equation}
    W_2^2(\mu, \nu) = \sup_{v\in  L^1(\nu)} \left\{ \int v^c(x) d\mu(x) + \int v(y) d\nu(y) \right\},
\end{equation}
where  functions $v$ is called a {potential} function, where $v^c(x) \coloneqq \min_{y \in \mathcal{Y}} \left[ \frac{1}{2} \|x - y\|^2 - \nu(y) \right]$ is the $c$-transform of the potential function.

In practice, single cell datasets often contain outliers, and hard constraints on marginals may result in undesired transport map. The unbalanced optimal transport addresses this issue by relaxing the mass conservation assumption, thereby avoiding problematic transport maps between the distributions. Specifically, it replaces hard constraints on marginals with a penalty term to measure the discrepancy between two measures. In this work, we focus on the scenario where only the constraint on the target measure $\nu$ is relaxed and define the corresponding unbalanced OT-based Wasserstein distance as 
\begin{equation}
W^2_{2,ub}(\mu, \nu) := \inf_{\pi \in \Pi(\mu)} \left[ \int_{\mathcal{X} \times \mathcal{Y}} \frac{1}{2} \|x - y\|^2 d\pi(x, y) + \tau D_{\psi}(\pi_1 \| \nu) \right],
\label{eq:unbalanced}
\end{equation}
where $\pi$ is the set of positive measures on the product space whose first marginal is $\mu$, and $D_\psi$ is the Csisz\`ar divergence which is defined as:
\[
D_\psi(\mu \| \nu) =
\begin{cases}
    \int_\Omega \psi\left(\frac{d\mu}{d\nu}\right) \, d\nu & \text{if } \mu \ll \nu, \\
    +\infty & \text{otherwise},
\end{cases}
\]
where $\psi: [0, \infty) \to \mathbb{R}$ is a convex, lower semi-continuous, and non-negative function with $\psi(1) = 0$, $\frac{d\nu}{d\mu}$ is the Radon-Nikodym derivative of $\nu$ with respect to $\mu$. The parameter $\tau$ controls the level of imbalance tolerance—referred to as the unbalanced tolerance in this paper—and as $\tau \rightarrow \infty$, the UOT reduces to the classical OT.  

This formulation implicitly assumes that the second marginal of $\pi$ is absolutely continuous with respect to $\nu$. We then introduce the semi-dual formulation of the above unbalanced optimal transport problem
\begin{equation}
    W_{2,ub}^2(\mu, \nu) = \sup_{v\in \mathcal{C}(\mathcal{Y})} \left\{ \int v^c(x) d\mu(x) + \int -\psi_\tau^*\left(-v(y)\right) d\nu(y) \right\}
    \label{eq:unbalanceddual}
\end{equation}
where the supremium is taken over the set of continuous functions over the domain $\mathcal{Y}$. To incorporate the unbalanced tolerance parameter $\tau$, we scale the convex conjugate and define $\psi_\tau^*(s)$ as follows: $\psi^*(s) = \sup_{t \in \mathbb{R}} \left\{ st - \tau\psi(t) \right\}$. According to \cite{choi2024generative}, any non-decreasing, convex, and differentiable function can be a candidate of $\psi_\tau^*$. In particular, when the convex conjugate $\psi_\tau^*(s) = s$, the formulation reduces to the classic OT problem. In this work, we adopt the Kullback-Leibler (KL) divergence as the choice of $\psi$, where $\psi(t) = t\log(t) - t + 1$. It's scaled convex conjugate admits the closed-form expression: $\psi^*_\tau(s) =\sup\limits_{t>0}\{st - \tau\cdot\psi(t)\} = \tau(e^{s/\tau} - 1), \forall s$.

\subsection{A robust local Fr\'echet regression for probability measures}
In single cell applications that we consider in this paper, we  are interested in learning how such cellular distribution $\nu$ changes over time based on the observed snapshots of the cells measured at $N$ different time points, $\nu_i, i=1, \cdots, N$. In the setting of regression analysis,   we are interested in estimating  the mean of probability measures {conditioned on time}. 
We consider the general setting of  local {Fr\'echet regression} \citep{petersen2019frechet} to estimate the conditional mean of general objects in a metric space: 
\[
\bar{\mu}(t) = \underset{\mu \in \mathcal{P}_{2,ac}(\mathcal{X})}{\arg\min} \frac{1}{N} \sum_{i=1}^N s_i(t,h) \, d^2(\mu, \nu_i),
\]
where $d(\cdot, \cdot)$ denotes the metric on the space of probability distributions, and $s_i(t,h)$ is a kernel-based weight encoding the relevance of subject $i$'s distribution at time $t_i$ to the target time $t$. The bandwidth $h > 0$ controls the degree of smoothing. We adopt the local linear weight form of \cite{petersen2019frechet}, where
\[
s_i(t,h) = \frac{1}{\hat{\sigma}_0^2} K_h(t_i - t)\left[\hat{\mu}_2 - \hat{\mu}_1 (t_i - t)\right],
\]
with $K_h(t_i - t) = h^{-1} K\left(\frac{t_i - t}{h}\right)$ being a scaled kernel function. The terms $\hat{\mu}_1$, $\hat{\mu}_2$, and $\hat{\sigma}_0^2$ are kernel-weighted empirical moments defined in \cite{petersen2019frechet}. Smaller values of $h$ assign more weight to time points closer to $t$, leading to more localized estimates.

When adopting the Wasserstein distance as the metric, the local Fr\'echet regression corresponds to a {weighted Wasserstein barycenter}. To enhance robustness, we specifically use the {unbalanced} Wasserstein distance, leading to the objective function that we want to optimize over the distribution $\mu_t$:
\begin{equation}
V(\mu_t) = \sum_{i=1}^N \alpha_i(t) \, W^2_{2,ub}(\mu_t, \nu_i),
\label{eq:UOTbarycenter}
\end{equation}
where $\alpha_i(t) \geq 0$ and $\sum_{i=1}^N \alpha_i(t) = 1$. In practice, we normalize the weights $(\alpha_1(t), \cdots, \alpha_N(t))$ 
$=(s_1(t,h)/N, \dots, s_N(t,h)/N)$ to satisfy the barycenter constraint. The resulting barycenter $\bar \mu_t=\argmin V(\mu_t)$ provides an estimate of the cell population distribution at time $t$. For the remainder of the paper, since we estimate each distribution separately at each time 
$t$, we drop the subscript 
$t$ for notational simplicity.

\section{A Fixed Point Approach to Local Fr\'echet Regression Estimation} \label{Fixed-point}

We propose a fixed-point algorithm to obtain  the robust local Fr\'echet regression estimate by solving the optimization problem \eqref{eq:UOTbarycenter}, extending the  previous results of using the fixed point iteration  to estimate the classical Wasserstein barycenter by \citet{alvarez2016fixed}.

Let $\nu_i$ denote the cell distribution corresponding to time  $t_i$ for $i=1,2,\dots,N$, and let $\mu$ denote the robust barycenter to be estimated by minimizing the objective function \eqref{eq:UOTbarycenter}, which involves the calculation of the unbalanced optimal transport problem. Consider the optimal solution $\pi^*$ of the unbalanced optimal transport problem~\eqref{eq:unbalanced}, which is a Borel measure defined on $\mathcal{X}\times \mathcal{Y}$. This solution admits a disintegration with respect to the first marginal, expressed as $\pi^*(\mathrm{d}x,\mathrm{d}y) 
  = \mu(\mathrm{d}x)\,\gamma^*_x(\mathrm{d}y)$, where  $ \gamma^*_x(\mathrm{d}y)$ is called the optimal conditional plan of $Y$ given $x$.
Utilizing these conditional plans, we define an ``average map" function $\bar{T}$ as:
\begin{equation}
    \bar{T}(\mu) \coloneq \mathcal{L}(\sum_{i=1}^N \alpha_iT^*_{\mu\to\nu_i}(\mu))
    \label{averge map}
\end{equation}
where the deterministic push-forward function $T^*$ is defined by the condition mean $T^*(x)=\int_Y y \,\gamma^*_x(\mathrm{d}y)$. When a deterministic map exists, the optimal conditional plan reduces to a delta measure, and the definition of the map aligns with the standard OT case. 

We show that the minimizer of ~\eqref{eq:UOTbarycenter} is a fixed point of the average map $\bar{T}$, i.e., $\sum_{i=1}^N \alpha_iT^*_{\mu\to\nu_i}(x) = x$ for $\mu-a.s.$ $x$. We then propose the iterative fixed-point algorithm for solving the robust local Fr\'echet regression:
\begin{equation}
    \mu_{n+1} = \bar{T}(\mu_n).
    \label{iterative algor}
\end{equation}
The following theorem shows that that objective function $V$ is non-increasing  during the fixed-point iterations. 

\begin{thm}
\label{thm:fix_point}
Let $\nu_i$, $i=1,2,\dots,N$ be a set of measures for which we estimate the robust Wasserstein barycenter. If $\mu \in \mathcal{P}_{2,\text{ac}}(\mathbb{R}^d)$, then
\[
    V(\mu) \geq V(\bar{T}(\mu)),
\]
where $\bar{T}$ is the average map as defined in equation (3). In particular, if $\mu$ is the robust Wasserstein barycenter, then it is a fixed point of the average map $\bar{T}$.
\end{thm}

\begin{remark*}
The result implies that the solution to the local Fr\'echet regression, the unbalanced Wassertein barycenter, can be characterized by the fixed-point of the average mapping function. The fixed-point iteration $\mu_{n+1} = \bar{T}(\mu_n)$ monotonically decreases the objective $V(\mu)$ and converges to a local minimum. Stronger convergence guarantees, such as weak convergence of $\mu_n$ to the fixed point, are difficult to establish under general case. This is largely due to the KL divergence term in the unbalanced Wasserstein distance, which complicates proving the continuity of $V(\mu)$ under weak convergence. Nevertheless, in practice, the proposed neural solver with proper initialization often exhibits favorable convergence behavior, even in high-dimensional and complex multi-modal settings (Figure~\ref{fig:supp_pretrain}).
\end{remark*}



\ignore{
\begin{thm}
The sequence $\{\mu_n\}$ defined by the iterative algorithm in equation (4) is tight. Under condition (2), every weakly convergent subsequence of $\mu_n$ must converge in $W_2$ distance to a probability measure in $\mathcal{P}_{2,\text{ac}}(\mathbb{R}^d)$ that is a fixed point of the map $G$. In particular, if $G$ has a unique fixed point $\bar{\mu}$, then $\bar{\mu}$ is the barycenter of $\nu_1, \dots, \nu_k$, and $W_2(\mu_n, \bar{\mu}) \to 0$.
\end{thm}

\begin{proof}[Sketch of proof]
The tightness of the sequence follows from the assumption that the space is compact. The convergence of every weakly convergent subsequence to a fixed point follows from arguments similar to those in Alvarez-Esteban et al. (2016).
\end{proof}
}

\section{Computational Algorithm Using Neural Network Modeling} 
\label{NeuralNet}

\subsection{Model reparameterization  using neural networks}

As discussed in the previous section, estimating the robust Wasserstein barycenter requires iteratively solving the unbalanced optimal transport problem between the current measure $\mu_n$ and the target measures $\nu_1, \dots, \nu_N$. However, there is no closed-form solution for the optimal coupling $\pi^*$ or the optimal conditional plan $\gamma^*_x$. To address this, we utilize a deep neural network to numerically solve the semi-dual form of the UOT problem \eqref{eq:unbalanceddual}. Following \citet{choi2024generative} and \citet{gazdieva2024robust}, we introduce a neural network $T_{\theta} : \mathcal{X} \times \mathcal{S} \to \mathcal{Y}$ parameterized by $\theta$, which approximates the convex conjugate of the potential function $v^c(x)$. Specifically, $\mathcal{S}$ denotes an auxiliary space that introduces stochasticity into the map, and $T_{\theta}$ defines a stochastic map satisfying: 
\begin{equation} T_{\theta}(x, s) \in \arg\inf_{y \in \mathcal{Y}} \left[ \frac{1}{2} |x - y|^2 - v(y) \right], \label{eq:stochasticmap} \end{equation} which yields $v^c(x) = \left[ \frac{1}{2} |x - T_\theta(x, s)|^2 - v(T_\theta(x, s)) \right]$.

By Theorem 2 of \citet{gazdieva2024robust}, the optimal conditional plan for the UOT problem is contained in the optimal saddle points, that is, $\gamma^*_x(y) \in \arg\inf_{y \in \mathcal{Y}} \left[ \frac{1}{2} |x - y|^2 - v(y) \right]$. Therefore, our stochastic map provides a valid parametrization of the optimal conditional plan.  Additionally, we introduce another neural network $v_\omega : \mathcal{Y} \to \mathbb{R}$ to parametrize the potential function $v(y)$. With these parametrizations, the semi-dual problem can be reformulated as: 
\begin{align}
\label{eq:parameterization}
W_{2,ub}^2(\mu, \nu_i)
&= \sup_{v_{\omega_i}}
\Biggl\{
  \int \inf_{T_{\theta_i}}
    \Bigl[
      \tfrac{1}{2}\lVert x - T_{\theta_i}(x, s)\rVert^2
      - v_{\omega_i}\bigl(T_{\theta_i}(x, s)\bigr)
    \Bigr] \,d\mu(x) \\
&\qquad\quad
  - \int \psi_\tau^*\bigl(-v_{\omega_i}(y)\bigr)\,d\nu_i(y)
\Biggr\}.
\nonumber
\end{align}

Finally, we introduce a neural network $G_\xi : \mathcal{Z} \to \mathbb{R}^D$ to parametrize the sequence of measures generated during the iterative algorithm \eqref{iterative algor}. Specifically, we represent the current measure as the push-forward $\mu = (G_\xi)_{\#}\mu_z$, where $\mu_z$ is a predefined latent distribution that is easy to sample from, such as the standard Gaussian $\mathcal{N}(0, I_{d_z})$ in the latent space $\mathcal{Z} \subseteq \mathbb{R}^{d_z}$.

\subsection{A fixed-point algorithm for regression estimation}
\label{sec:fixedpoint}

Equipped with the above neural network parameterizations, we are ready to implement the fixed-point iterative algorithm \eqref{iterative algor}, presented as Algorithm \eqref{alg:fixed_point}. Following \citet{korotin2022wasserstein}, we approximate the average map $\bar{T}$ by regressing $(G_\xi){\#} \mu_z$ onto $\bar{T}((G_{\xi_0}){\#} \mu_z)$, where $G_{\xi_0}$ is a fixed copy of $G_\xi$. The pushforward \( T_i^*(x) = \mathbb{E}_{\gamma_{i,x}^*}[y] \) is approximated by averaging over the latent variable \( s \) through the stochastic map \( T_\theta(x, s) \). Specifically, we update the generator parameter $\xi$ by minimizing the following loss:
\[
\int_{\mathcal{Z}} \ell \left( G_\xi(z), \sum_{i=1}^N \alpha_i T_{\theta_i}\left( G_{\xi_0}(z) \right) \right) d\mu_z,
\]
where $\ell$ denotes a properly chosen loss function, such as the mean squared error (MSE).

In practice, since the true measure is not accessible, we approximate the integral using Monte Carlo sampling. All networks, including $T_{\theta_n}, v_{\theta_n},$ and $G_\xi$, are implemented as simple fully connected neural network layers. The loss functions are optimized using the stochastic gradient descent algorithm Adam. Throughout this paper, we implement $T_\theta(x, s)$ as a deterministic map $T_\theta(x)$, which empirically already yields satisfactory performance.

\begin{algorithm}[h!]
\caption{Fixed point algorithm for unbalanced barycenter}
\label{alg:fixed_point}
\KwIn{Latent $\mathcal{Z} \subseteq \mathbb{R}^{d_z}$; Auxiliary $\mathcal{S} \subseteq \mathbb{R}^{d_s}$; Empirical measures $\widehat{\mathbb{P}}_1, \dots, \widehat{\mathbb{P}}_N \in \mathbb{R}^D$; weights $\alpha_1, \dots, \alpha_N > 0$ $(\sum_{n=1}^N \alpha_n = 1)$; generator $G_\xi : \mathbb{R}^{d_z} \to \mathbb{R}^D$; mapping networks $T_{\theta_1}, \dots, T_{\theta_N}: \mathbb{R}^D \to \mathbb{R}^D$; potentials $\nu_{\omega_1}, \dots, \nu_{\omega_N}: \mathbb{R}^D \to \mathbb{R}$; regression loss $\ell : \mathbb{R}^D \times \mathbb{R}^D \to \mathbb{R}^+$; unbalance tolerance $\tau$; convex conjugate of the divergence function $\Psi_\tau^*: \mathbb{R} \to \mathbb{R}$; number of iterations per network: $K_G, K_T, K_\nu$; Batch size $B$.}
\KwOut{Generator satisfying $G_\xi \sharp \mathcal{S} \approx \bar{\mathbb{P}}$; OT maps satisfying $T_{\theta_n} \sharp(G_\xi \sharp \mathcal{Z}) \approx \mathbb{P}_n$.}

\Repeat{converged}{
    \#\textit{ Update OT maps} $\&$ \textit{potentials}
    
    \For{$k_\nu = 1, 2, \dots, K_\nu$}{
        \For{$k_T = 1, 2, \dots, K_T$}{
            Sample batches $Z \sim \mathcal{Z}; s\sim \mathcal{S}; X \gets G_\xi(Z)$ \;
            \#\textit{ Joint training} 
            
            $\mathcal{L}_T \gets \frac{1}{N} \sum_{i=1}^N \left\{  \frac{1}{|X|} \sum\limits_{x \in X} \left[ \frac{1}{2} \| x - T_{\theta_i}(x,s) \|^2 - \nu_{\omega_n}(T_{\theta_i}(x,s)) \right]\right\}$\;
            Update $\theta_1, ..., \theta_n$ by descending $\mathcal{L}_T$\;
        }
        Sample batches $Z \sim \mathcal{Z}; s\sim \mathcal{S}; X \gets G_\xi(Z); Y_i \sim \widehat{\mathbb{P}}_i, i = 1,\dots,N$\;
        \ignore{
        $\mathcal{L}_\nu \gets \frac{1}{N} \sum\limits_{i=1}^N \left[\frac{1}{|X|} \sum\limits_{x \in X} \psi^*_\tau\left(-\| x - T_{\theta_i}(x,s) \|^2+\nu_{\omega_i}(T_{\theta_i}(x,s))\right) - \frac{1}{|Y|} \sum\limits_{y_i \in Y_i} \nu_{\omega_i}(y_i)\right]$\;}

        $\mathcal{L}_\nu \gets \frac{1}{N} \sum\limits_{i=1}^N \left[\frac{1}{|X|} \sum\limits_{x \in X} \nu_{\omega_i}(T_{\theta_i}(x,s)) + \frac{1}{|Y|} \sum\limits_{y_i \in Y_i} \psi^*_\tau\left(-\nu_{\omega_i}(y_i)\right)\right]$\;
        
        Update $\nu_{\omega_1}, \dots, \nu_{\omega_N}$ by descending $\mathcal{L}_\nu$\;
    }

    \#\textit{ Update the generator}
    
    \For{$k_G = 1, 2, \dots, K_G$}{
        Sample batch $Z \sim \mathcal{Z}$\;
        $\mathcal{L}_G \gets \frac{1}{|Z|} \sum_{z \in Z} \ell \left(G_\xi(z), \sum_{n=1}^N \alpha_i T_{\theta_i}(G_{\xi_0}(z))\right)$\;
        Update $\xi$ by descending $\mathcal{L}_G$\;
    }
}
\end{algorithm}

Compared to the algorithm for computing the unbalanced Wasserstein barycenter proposed in \citet{gazdieva2024robust}, our fixed-point algorithm—though requiring tri-level optimization and being less efficient when the number of distributions is small—enjoys two key advantages: (1) the computational complexity of our algorithm scales linearly with the number of distributions $N$, whereas the method in \cite{gazdieva2024robust} scales quadratically in $N$ due to the congruence condition. This gives our approach an advantage when dealing with large numbers of distributions. (2) The fixed-point algorithm provides an explicit generator $G_\xi$ for the barycenter, enabling direct sampling of new points from the barycenter. In contrast, the method in \cite{gazdieva2024robust} relies on sampling from the target distributions and mapping the samples back to the barycenter.

\subsection{Addressing the mode-collapse in complex datasets}

We observe that in complex single-cell datasets, gene expression often exhibits a strongly multimodal pattern, which leads to significant convergence issues for the fixed-point approach. Although UOT mitigates this problem by learning a more conservative map, relying solely on UOT to address convergence can result in substantial loss of accuracy in recovering the true barycenter. In particular, setting $\tau \to 0$ guarantees convergence but leads to a trivial solution — the identity map.

Previous studies \citep{arjovsky2017wasserstein, nowozin2016f} have drawn connections between generative adversarial networks (GANs) and the dual formulations of both classical and unbalanced optimal transport problems. Given the similarity between the training objectives of GANs and UOT, we interpret the potential function in our framework as analogous to the discriminator in GANs. Mode collapse, a well-known phenomenon where the generator fails to capture all data modes, has attracted extensive research. Several explanations have been proposed, including large gradients in the discriminator \citep{kodali2017convergence}, generator instability \citep{arjovsky2017towards}, and catastrophic forgetting \citep{thanh2020catastrophic}. Drawing on this analogy, we suspect that our framework is also prone to mode collapse. Specifically, we observe that when the generator fails to adequately cover all modes of the data, the potential functions will continually increase the "scores" assigned to the unreached modes. As a result, the learned transport map fails to represent the true barycenter meaningfully.

To address the mode collapse issue, we propose to pretrain the generator $G_\xi$ to approximate the empirical distribution of the entire dataset before estimating the barycenter. This step ensures that the generator captures the diverse modalities inherent in complex single-cell datasets, which is critical for improving the convergence behavior of the fixed-point algorithm. Compared to GAN-based training algorithms, machine learning methods such as normalizing flows and diffusion models are known for their effectiveness in capturing multiple modes within a dataset \citep{xiao2021tackling}.
In this paper, we adopt the normalizing flow variational autoencoder (VAE-NF) framework, as originally introduced by \cite{rezende2015variational}, to pretrain our generator $G_\xi$. Specifically, the generator is implemented as the decoder in the VAE architecture, and we optimize the model by minimizing the flow-based free energy bound computed from the flow. This allows the generator to effectively capture all data modes when sampling from a standard Gaussian latent space. Additionally, we incorporate $\beta$-annealing, gradually increasing the weight of the KL divergence term during training to further improve performance of pretraining.  The detailed pre-training algorithm is presented in Algorithm \ref{alg:VAE-NF}.

\begin{algorithm}[h!]
\caption{Pre-training of the generator with VAE-NF}
\label{alg:VAE-NF}
\KwIn{
    Latent space $\mathcal{Z} \subseteq \mathbb{R}^{d_Z}$; Empirical measures $\widehat{\mathbb{P}}_1, \dots, \widehat{\mathbb{P}}_N \in \mathbb{R}^D$; 
    Encoder $E_\phi: \mathbb{R}^D \rightarrow \mathbb{R}^{d_Z}$;
    Decoder $G_\xi: \mathbb{R}^{d_Z} \rightarrow \mathbb{R}^{D}$; 
    Normalizing flow $f_K \circ f_{K-1} \circ \dots \circ f_1$, where $f_k: \mathbb{R}^{d_Z} \to \mathbb{R}^{d_Z}$; Beta coefficient $C_\beta$; Number of epochs $N$; Batch size $B$. }
\KwOut{Pretrained generator $G_\xi$}

\BlankLine
\textbf{Step 1: Pooling.} Form $\mathbb{P}_{data} = \frac{1}{N} \sum_{n=1}^N \widehat{\mathbb{P}}_n$\;

\textbf{Step 2: Training Loop} \\
\For{$epoch = 1, 2, \dots, N$}{
    Update $\beta \leftarrow \frac{epoch}{C_\beta \cdot N}$; Sample mini-batch $X \sim \mathbb{P}_{data}$\;
    Encode: $\mu, \sigma \leftarrow E_\phi(X)$; \\
    Sample latent: $Z_0 \sim \mu + \sigma \cdot \epsilon$, $\epsilon \sim \mathcal{N}(0, I)$; \\
    Flow transform: $Z_K, \log \left.\bigl|\det\!\bigl(\tfrac{\partial f_k}{\partial \mathbf{Z}}\bigr)\bigr|\right|_{\mathbf{Z} = \mathbf{Z}_k} \leftarrow f_K \circ f_{K-1} \circ \dots \circ f_1(Z_0)$; \\
    Decode: $\hat{X} \leftarrow G_\xi(Z_K)$;

    Compute reconstruction loss: 
    \[
    \mathcal{L}_{recon} = \frac{1}{B} \sum_{i=1}^B \| X^{(i)} - \hat{X}^{(i)} \|^2
    \]

    Compute KL divergence via flow-based free energy bound:
    \[
    \mathcal{L}_{kl} = \frac{1}{B} \sum_{i=1}^B \left( 
        -\frac{1}{2} \frac{\|Z_0^{(i)}-\mu^{(i)}\|^2}{\left(\sigma^{(i)}\right)^2} - \frac{S}{2}\log\sigma^{(i)}
        +\frac{1}{2} \|Z_K^{(i)}\|^2 + \sum_{k=1}^K \log \left.\bigl|\det\!\bigl(\tfrac{\partial f_k}{\partial \mathbf{Z}}\bigr)\bigr|\right|_{\mathbf{Z} = \mathbf{Z}_k^{(i)}} 
    \right)
    \]

    Total loss: $\mathcal{L} = \mathcal{L}_{recon} + \beta \mathcal{L}_{kl}$\; Update $\phi, \xi$ by descending the total loss $\mathcal{L}$;
}
\Return{Pretained generator $G_\xi$.}
\end{algorithm}

Pretraining the generator significantly improves the convergence behavior of the fixed-point algorithm in complex multimodal settings. We illustrate this using a COVID-19 scRNA-seq dataset (see Supplement~\ref{supp:COVID}) in Figure~\ref{fig:supp_pretrain}, where the goal is to compute the unbalanced barycenter of COVID-19 subjects. We compare two initialization strategies: the standard approach, where the generator is initialized to mimic a Gaussian, and our proposed method using VAE-NF pretraining, which approximates the empirical distribution of the whole dataset.

The loss curves in the top panel of Figure~\ref{fig:supp_pretrain} show that VAE-NF pretraining leads to stable convergence of the generator during fixed-point updates. In the bottom row, the learned barycenter under VAE-NF initialization captures the underlying structure of the COVID-19 cell population accurately. In contrast, the model with standard initialization fails to converge, resulting in a generator that poorly represents the data.

\subsection{Estimating transport maps}
\label{sec:transportmap}

After performing distributional interpolation using robust Fr\'echet regression, we can reconstruct cellular developmental trajectories or identify signature genes that drive changes in cell abundance or transcriptomic profiles. To reconstruct such trajectories, we need to estimate the transport plan between adjacent distributions. Instead of using the input convex network \citep{amos2017input}, we approach this by solving the unbalanced optimal transport semi-dual problem~\eqref{eq:unbalanceddual}, using the same parameterization technique introduced in Equation~\eqref{eq:parameterization}. The resulting neural network $T_{\theta}(x,s)$ represents the optimal transport plan between distributions at different time points. Since we implement it as a deterministic map in this paper, $T_{\theta}(x)$ is automatically a transport map. Given the inherent noise in single-cell datasets, the unbalanced optimal transport formulation offers improved robustness to outliers, allowing for the extraction of clearer biological signals compared to its balanced counterpart.

\section{Simulation Study}
\ignore{
We highlight two key advantages of our algorithm over traditional approaches. First, single-cell datasets often contain outlier cells—even after quality control—which can degrade the quality of learned transport maps. Our robust local Fréchet regression framework is inherently resistant to such outliers, enhancing its reliability for single-cell analysis. Second, the pretraining of the generator improves the convergence behavior of the fixed-point updates, especially in complex multimodal settings. In this section, we provide empirical evidence using both synthetic and real datasets to demonstrate these advantages.

\subsection{UOT increase the robustness}
}

We demonstrate the robustness of our proposed robust local Fr\'echet regression by comparing it with its balanced OT counterpart under a synthetic dataset. Single-cell data often contain outlier cells, even after quality control, which can severely distort transport-based representations. 
To evaluate this, we construct synthetic datasets based on ten-dimensional Gaussian mixtures, where each component mimics a cell type in the scRNA-seq dataset. The mixture consists of four central components representing major cell types (comprising 95\% of the total mass), and four peripheral components simulating outlier cells (5\% of the mass). Further details of the simulation can be found in Supplement~\ref{supp:simulation}. 

Our approach leverages unbalanced optimal transport  to mitigate the sensitivity to these outliers, offering a more reliable estimate of the underlying population structure.
Figure~\ref{fig:simu_outlier} compares the barycenters computed by our unbalanced method (left) and the classical balanced barycenter (right). The unbalanced barycenter faithfully recovers the structure of the four main modes while being robust to the outliers. In contrast, the classical barycenter is more dispersed and distorted, as it is pulled toward the outlier components.

We also compare the learned transport plans in both settings (Figure~\ref{fig:supp_simu}). In the unbalanced case, transport plans successfully captured mapping relationship between barycenter and marginal distribution, yielding accurate reconstructions. In contrast, the transport in the balanced case is heavily influenced by outliers, resulting in distorted mappings and a poor recovery of the underlying distribution.

\begin{figure}[H]
    \centering
    \includegraphics[width=0.9\linewidth]{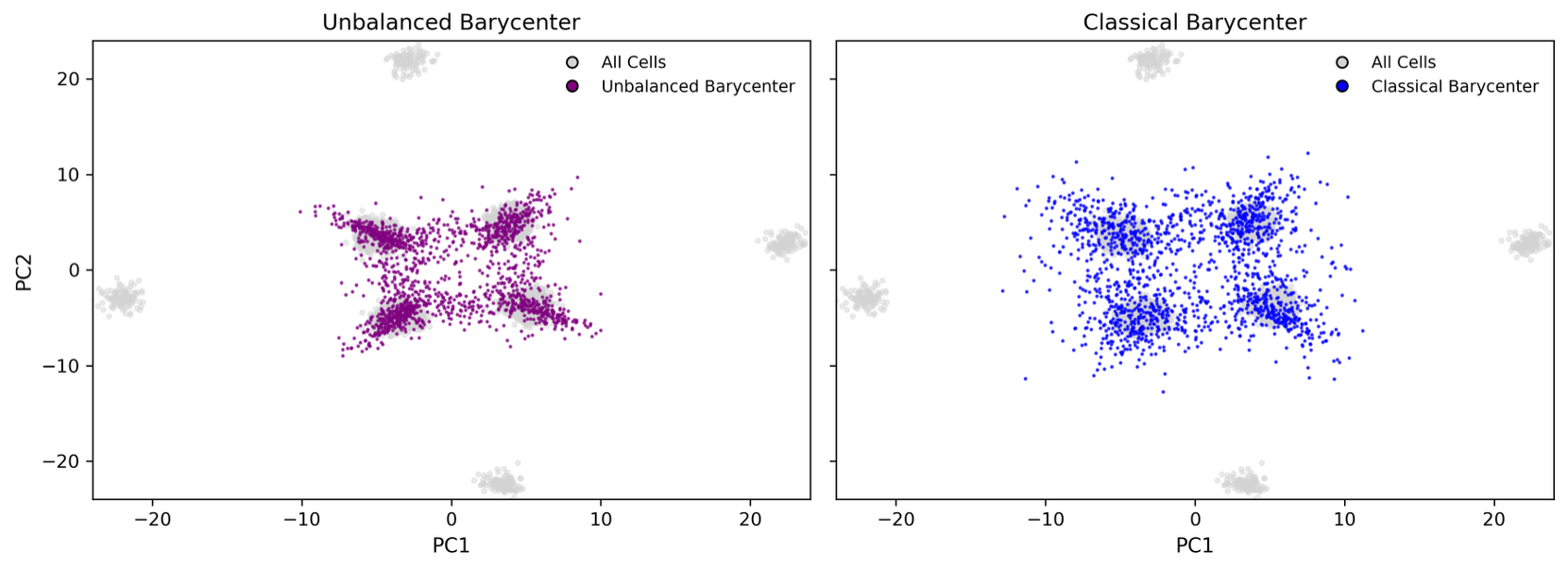}
    \caption{Comparison between the unbalanced OT and classical OT barycenters in the presence of outliers, visualized using the top two principal components. Grey points represent pooled samples from ten Gaussian mixtures.}
    \label{fig:simu_outlier}
\end{figure}

\section{Applications To Cell Dynamic Analysis Based on scRAN-seq Data} \label{Realdata}

We demonstrate the effectiveness of our robust local Fr\'echet regression framework for analyzing cellular dynamics using three temporal scRNA-seq datasets. Our method enables interpolation of cellular distributions at any time point based on the scRNA-seq data observed at  a few time points, facilitating the comparison of the cell distributional changes over time. The resulting interpolated distributions across additional time points can also be leveraged to improve the temporal resolution of reconstructed cellular dynamics, thereby facilitating the identification of distinct differentiation pathways and signature genes.\\

\noindent {\bf Embryoid data - human embryonic stem cell differentiation dataset} \citep{moon2019embryoid}. 
The dataset profiles low-passage H1 human embryonic stem cells (hESCs) undergoing differentiation as embryoid bodies (EBs) over a 27-day time course. Approximately 31,000 cells are profiled, and 16,825 cell measurements are left for analysis after preprocessing. Samples are collected at 3-day intervals and pooled for measurement on the 10x Genomics Chromium platform, resulting in a dataset comprising five time points: "Day00-03", "Day06-09", "Day12-15", "Day18-21", and "Day24-27" ($N=5$).  We extract the top PCs for our local Fr\'echet analysis following the preprocessing steps in \citet{huguet2022manifold}.

\medskip \noindent {\bf Statefate data - mouse hematopoietic stem and progenitor cell differentiation  dataset} \citep{weinreb2020statefate}. The dataset profiles mouse hematopoietic stem and progenitor cells (HSPCs) undergoing in vitro differentiation. Approximately 130,000 single-cell transcriptomes are collected at three time points: Day 2 (50\%), Day 4 (30\%), and Day 6 (remaining cells). Single-cell RNA sequencing combined with genetic barcoding enables simultaneous profiling of cell states and clonal fates. We use the preprocessed data with extracted  PCs following the steps described in \citet{bunne2023learning}.

\medskip\noindent {\bf Reprogramming data - mouse embryonic fibroblasts reprogramming dataset} \citep{schiebinger2019optimal}. The dataset profiles the reprogramming of mouse embryonic fibroblasts (MEFs) into induced pluripotent stem cells (iPSCs) over an 18-day time course. A total of 259,155 single-cell transcriptomes are collected across 39 time points and sampled every 12 hours (and every 6 hours between days 8 and 9). Single-cell RNA sequencing is performed using the 10x Genomics Chromium platform. We download the processed and log(TPM+1)-transformed dataset from \url{https://portals.broadinstitute.org/single_cell}. We select the top highly variable genes (HVGs) and computed the principal components  for analysis.

\subsection{Implementation details}
\label{sec:implementation}

Guided by the sensitivity analysis in Section~\ref{sec:sensitivity}, in our benchmark experiments on temporal datasets, we set: (1) the bandwidth $h = 3$ for the Embryoid dataset, $h = 2$ for the Statefate dataset, and $h = 0.5$ for the Reprogramming dataset; (2) the unbalanced tolerance $\tau = 5$ for the Embryoid and Reprogramming datasets, and $\tau = 1$ for the Statefate dataset. For the Reprogramming dataset, due to the large number of time points, we keep only those distributions with weights greater than 0.01. This not only reduces computational burden but also improves prediction accuracy, as including time points far from the target may degrade the performance.

In our benchmark experiment, we uses fully connected neural networks with 4 hidden layers and width 196 for $T_{\theta_n}$ $v_{\omega_n}$. a learning rate 0.0003. They are optimized with Adam (learning rate = 0.0003, weight decay = 1e-10) and a batch size of 64. The generator $G_\xi$ is implemented as a 4-layer fully connected network with 256 hidden units per layer, optimized using Adam (learning rate = 0.0001, weight decay = 1e-8) with a batch size of 128. The number of iterations for each update step, as defined in Algorithm~\ref{alg:fixed_point}, is $K_G=50$, $K_T=10$, and $K_v=50$. We train the UOT/OT models for 25, 35, and 25 epochs on the Embryoid, Statefate, and Reprogramming datasets, respectively. We train all models in the sensitivity analysis by 25 epochs.

We follow the ``hold-out" evaluation used in MioFlow~\citep{huguet2022manifold} to ensure a fair comparison. Specifically, we use their geodesic autoencoder (GAE) with $\alpha$-decay distance to embed the PC data. The encoder architecture is set to [features\_dim, 200, 200, gae\_embedded\_dim], while the neural network learning the flow trajectory has three hidden layers of size [64,64,64]. All other hyperparameters follow the default settings in their python tutorial \url{https://github.com/KrishnaswamyLab/MIOFlow/blob/main/notebooks/%5BTutorial%5D%20EB-Hold-out.ipynb}. Similarly, for the Reprogramming dataset, we use only the adjacent time points, $t-2$, $t-1$, $t+1$, and $t+2$, to predict the distribution at time $t$, as including too many time points can negatively impact the performance.

Finally, to estimate the transport plans,  we use 4-layer fully connected neural networks with width 196 to parametrize the transport plans and potential functions. The bi-level min-max optimization problem is solved using the first step of Algorithm~\ref{alg:fixed_point}. To ensure accuracy, we set $K_v=500$ and $K_T=100$ with a batch size of 128. In all applications, we set the unbalanced tolerance parameter $\tau=5$.

\subsection{Benchmark of distributional interpolations}

We first evaluate the performance of robust local Fr\'echet regression in accurately interpolating cellular distributions in a 20-dimensional PC space using the three temporal scRNA-seq datasets listed above.  Specifically, we hold out cells from one time point and predict the distribution of that time point using cells from the remaining time points.  We compute three distance-based metrics that quantify discrepancies between the predicted and ground truth distributions: maximum mean discrepancy (MMD), Earth Mover's Distance (EMD, equivalent to Wasserstein-1), and Wasserstein-2 distance. 

We benchmark our method to Mioflow \citep{huguet2022manifold}, which has demonstrated good performance in interpolating cellular trajectories. Additionally, we compare our method against two baseline methods and the classical optimal transport (OT) formulation of our algorithm, which corresponds to the case where the unbalanced tolerance $\tau \to \infty$. The first baseline method estimates the distribution at time $t$ by taking the midpoint between the cells at times $t-1$ and $t+1$. The second baseline method does not perform interpolation but instead reports, for each metric, the average of the distances between the ground truth distribution at time $t$ and its two adjacent time points. It reflects the error when directly using adjacent time point to predict time $t$. More implementation details can be found at Section~\ref{sec:implementation}.

Tables~\ref{benchmark:embryoid} presents the leave-one-out benchmark results across selected time points for the  three datasets, respectively.  For the  Embryoid dataset, the robust local Fr\'echet regression (UOT) and its classical OT formulation achieve comparable performance and consistently outperform MioFlow and the two baseline methods. Prediction errors of Mioflow increase at later time points, likely due to error accumulation when inferring the entire trajectory sequentially. The unbalanced version (UOT) outperforms the OT formulation at $t = 7.5$, but performs slightly worse at $t = 13.5$ and $t = 19.5$. This aligns with the observation that $t = 7.5$ is farther from its adjacent time points (as reflected by baseline method 2), making distribution prediction more challenging. In such cases, the unbalanced formulation provides an advantage by producing more robust and conservative estimates. 

For  the Statefate and the Reprogramming datsets, the robust local Fr\'echet regression consistently outperforms other methods. Overall, we observe that UOT-based regression estimates the cell distributions better than using the standard OT. 

We also visualize representative held-out time points from each dataset in Figure~\ref{fig:benchmark}. Local regression with both the UOT and OT methods accurately capture the finer structures of the true distributions, with UOT performing particularly well on the Statefate dataset. In contrast, while Mioflow and the baseline methods can successfully capture the main mode of the data, they are less accurate in preserving finer distributional details. Supplementary Figure~\ref{fig:supp1} and Figure~\ref{fig:supp2} present the full visualizations for the Embryoid and Reprogramming datasets, respectively. We observe that UOT and its balanced OT counterpart consistently outperform the other methods, while the UOT formulation exhibits increased robustness to outliers, particularly in the Reprogramming dataset.

\begin{table}[h!]
    \caption{Benchmark results on the 20-dimensional Embryoid dataset at $t = 7.5$ (Day 6–9), $t = 13.5$ (Day 12–15), and $t = 19.5$ (Day 18–21), Statefate data at $t=4$, and the Reprogramming data at $t=10, 13.5$ and $16.5$, 
    comparing interpolation accuracy across methods. Metrics include Maximum Mean Discrepancy (MMD), Earth Mover’s Distance (EMD), and Wasserstein-2 ($W_2$) distance, averaged over 100 repetitions.}
    \vspace{5pt}
    \label{benchmark:embryoid}
    \centering
    \begin{tabular}{l|ccc|ccc|ccc}
         \toprule
         \multicolumn{1}{c}{} & 
 \multicolumn{9}{c}{Embryoid data set}\\
        & \multicolumn{3}{c|}{$t=7.5$} & \multicolumn{3}{c|}{$t=13.5$} & \multicolumn{3}{c}{$t=19.5$} \\
        \cmidrule(lr){2-4} \cmidrule(lr){5-7} \cmidrule(lr){8-10}
        &  MMD(G) & EMD & $W_2$ & MMD(G) & EMD & $W_2$ & MMD(G) & EMD & $W_2$ \\
        \midrule
        Baseline 1  & 0.381 & 5.370 & 5.558 & 0.257 & 5.353 & 5.607 & 0.264 & 5.563 & 5.810 \\
        Baseline 2  & 0.483 & 6.482 & 7.024 & 0.203 & 5.922 & 6.350 & 0.183 & 5.958 & 6.237 \\
        Mioflow     & \textbf{0.257} & 5.332 & 5.514 & 0.158 & 6.028 & 6.520 & 0.156 & 6.900 & 7.310 \\
        OT          & 0.280 & 5.260 & 5.464 & 0.114 & \textbf{5.076} & \textbf{5.319} & 0.153 & \textbf{5.390} & \textbf{5.629} \\
        UOT         & 0.281 & \textbf{5.046} & \textbf{5.223} & \textbf{0.111} & 5.121 & 5.424 & \textbf{0.135} & 5.433 & 5.717 \\
        \bottomrule
        \multicolumn{1}{c}{} &
        \multicolumn{3}{c}{Statefate data}\\
        & \multicolumn{3}{c}{$t=4$} \\
        \cmidrule(lr){2-4}  
        &  MMD(G) & EMD & $W_2$ \\
        \midrule
        Baseline 1     & 0.235 & 6.743 & 7.138 \\
        Baseline 2     & 0.297 & 7.977 & 8.934 \\
        Mioflow        & 0.200 & 7.443 & 7.934 \\
        OT             & 0.116 & 6.219 & 6.622 \\
        UOT            & \textbf{0.086} & \textbf{5.672} & \textbf{6.114} \\
        \bottomrule
        \multicolumn{1}{c}{} &
        \multicolumn{9}{c}{Reprogramming data}\\
            & \multicolumn{3}{c|}{$t=10$} & \multicolumn{3}{c|}{$t=13.5$} & \multicolumn{3}{c}{$t=16.5$} \\
            \cmidrule(lr){2-4} \cmidrule(lr){5-7} \cmidrule(lr){8-10}
            & MMD(G) & EMD & $W_2$ 
            & MMD(G) & EMD & $W_2$ 
            & MMD(G) & EMD & $W_2$ \\
            \midrule
            Baseline 1   & 0.224 & 4.879 & 5.033  & 0.259 & 6.258 & 6.500  & 0.198 & 5.203 & 5.814 \\
            Baseline 2   & 0.128 & 5.059 & 5.307  & 0.118 & 5.612 & 6.002  & 0.064 & 4.324 & 4.933 \\
            Mioflow     & 0.248 & 5.324 & 5.585  & 0.064 & 5.699 & 5.955  & \textbf{0.033} & 4.605 & 5.082 \\
            OT          & 0.092 & 4.446 & 4.621  & 0.087 & 4.950 & 5.170  & 0.051 & 4.017 & 4.552 \\
            UOT         & \textbf{0.079} & \textbf{4.389} & \textbf{4.559}  & \textbf{0.086} & \textbf{4.868} & \textbf{5.100}  & 0.044 & \textbf{3.948} & \textbf{4.481} \\
            \bottomrule
        \end{tabular}
\end{table}

\begin{figure}[h]
    \centering
    \includegraphics[width=\linewidth]{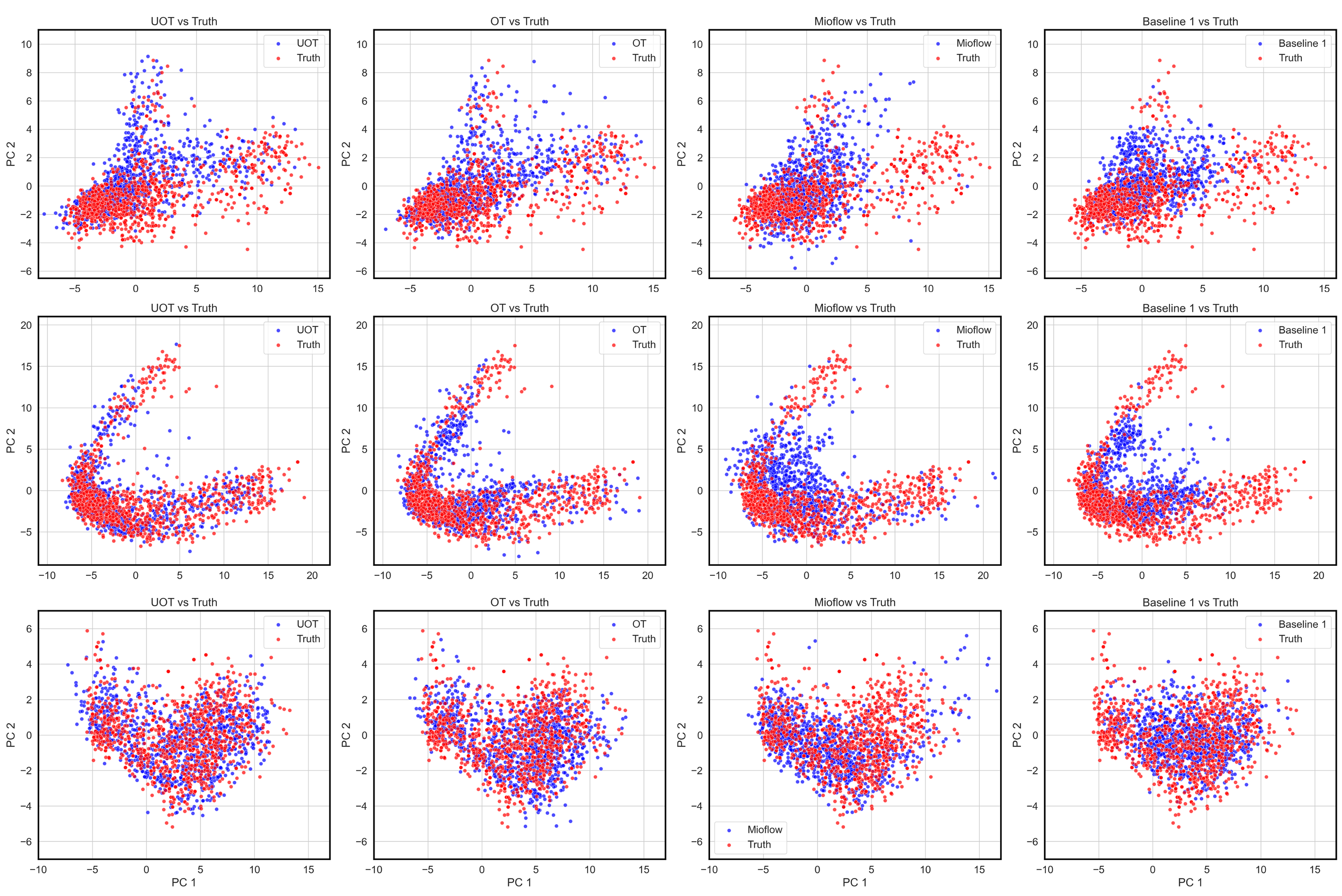}
    \caption{Comparison of the predicted cellular distributions (blue) against the ground truth (red) visualized on the top two principal components (PC1 and PC2) across different methods and datasets. Visualizations are shown for representative held-out time points in each dataset. \textbf{First row:} Embryoid dataset at \(t = 13.5\); \textbf{Second row:} Statefate dataset at \(t = 4\); \textbf{Third row:} Reprogramming dataset at \(t = 10\).}
    \label{fig:benchmark}
\end{figure}

Figure~\ref{fig:trajectories}b illustrates two representative interpolations using our method on the Embryoid dataset. In both cases, the interpolated distribution is a smooth  transition between adjacent time points. This demonstrates the model’s ability to capture gradual distributional shifts over time and to interpolate realistic intermediate transcriptomics states.

\subsection{Recovering time trajectory of individual cells}
\label{sec:recovertraj}

Inferring cellular trajectories in temporal biological processes from initial transcriptomic states is of great interest in scRNA-seq studies. While the primary goal of the robust local Fr\'echet regression is to interpolate cellular distributions, the model also allows to reconstruct the underlying cellular trajectory by learning and composing transport maps between the  distributions, see section~\ref{sec:transportmap} for details.

\begin{figure}[h]
    \centering
    \includegraphics[width=\linewidth]{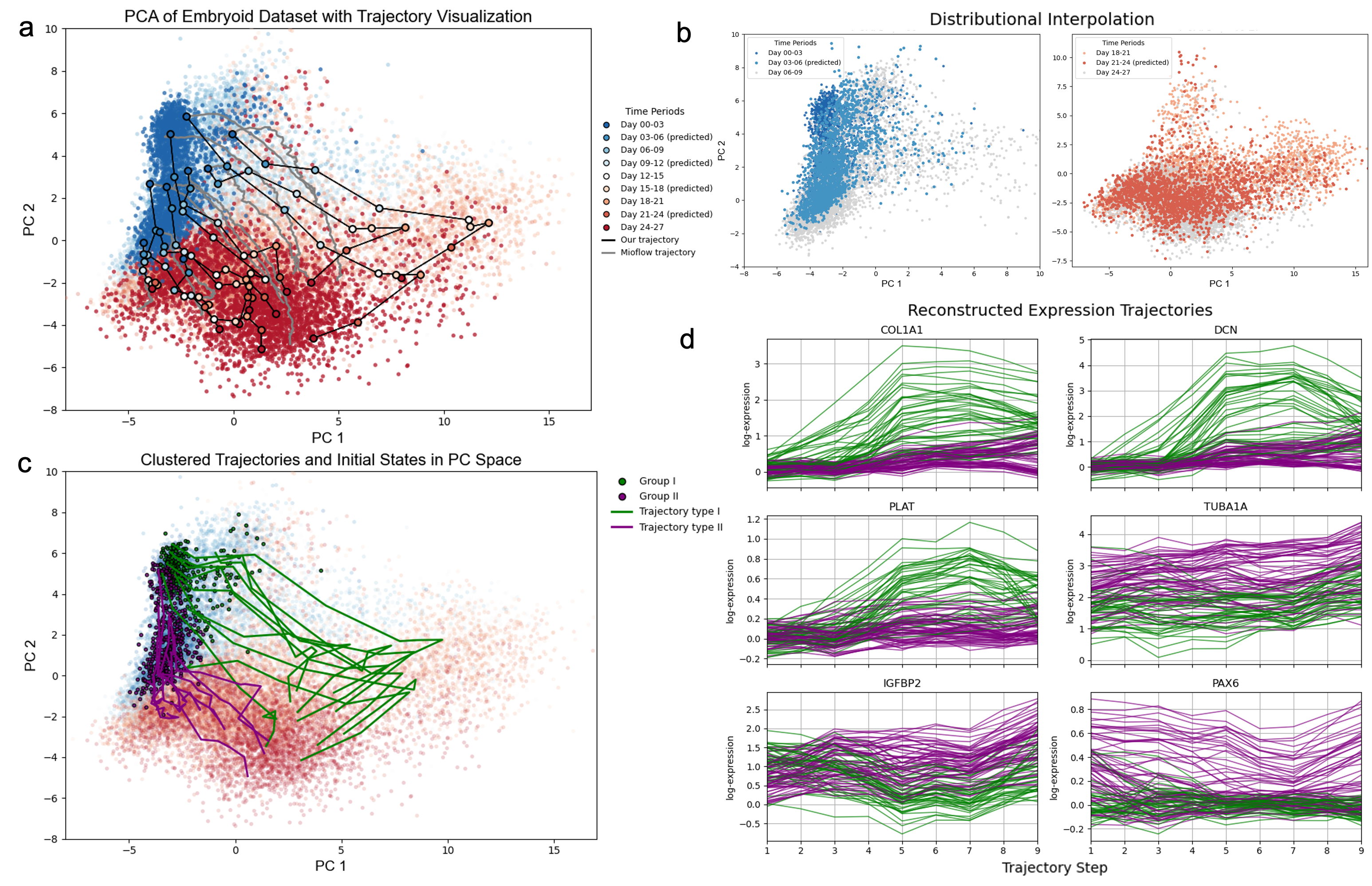}
    \caption{Distributional interpolation and trajectory reconstruction in the Embryoid dataset, visualized using the top two principal components (PC1 and PC2). \textbf{a,} Reconstructed trajectories from Day 00–03 to Day 24–27. Black lines trace our predicted trajectories; Each dot on them represents the predicted position of a cell in transcriptomic space at a given time point. The grey lines correspond to Mioflow predictions for the same cells.  \textbf{b,} Visualization of the interpolated distribution on Day 03-06 and Day 21-24. \textbf{c,} Two differentiation fates identified via Dynamic Time Warping (DTW) clustering. Initial states are colored by group (Group I: green; Group II: purple). \textbf{d,} Reconstructed gene expression trajectories for genes showing significant differences in intercept or slope across trajectory types.}
    \label{fig:trajectories}
\end{figure}

We show a few representative trajectories inferred by our method and compare them to those inferred by Mioflow  for the Embryoid dataset (Figure~\ref{fig:trajectories}a). Both methods capture similar global patterns of cellular developmental dynamics. Mioflow  maintains a smoother trajectory across time points. However, our method aligns with the observed and predicted distributions more precisely, effectively capturing finer structures and subtle transitions within the data. This suggests that while Mioflow provides a continuous trjectory  estimate, our method enhances fidelity in reconstructing the temporal dynamics of cellular states.

\subsection{Identifications of Trajectory associated signature genes}
\label{sec:trajanalysis}

After obtaining the reconstructed trajectories, we conduct downstream analyses to uncover distinct trajectory patterns and their associated genes. For each individual cell, the learned transport maps can be applied to predict its dynamic behavior over time. To identify distinct differentiation pathways, we calculate pairwise similarities between trajectories based on dynamic time wrapping (DTW) \citep{sakoe2003dynamic}, followed by clustering based on the resulting distance matrix.

To enable gene-level interpretation, we map the data from the PC space back to the original gene expression space. Specifically, the approximate reconstruction of gene expression is given by: $\hat{X} = ZW^T$, where $W\in\mathbb{R}^{p\times d}$ is the loading matrix and $Z\in\mathbb{R}^{n\times d}$ is the PC matrix. This linear reconstruction captures the gene expression encoded in the top $d$ principal components. Since PCA in Scanpy is internally centered, we add back the mean gene expression values to the reconstructed profiles to obtain estimates on the original scale.

To analyze gene expression trajectories, we fit a linear mixed effects model (LMM) for each gene, assessing whether expression patterns differ across trajectory types. The model includes a fixed effect for group (trajectory type), a continuous covariate for time, and their interaction, along with a random intercept for each trajectory cluster. Formally, we fit the following model:
\[
\texttt{Expr} \sim \texttt{Group} \times \texttt{Time} + (1|\texttt{traj\_id})
\]
where \texttt{Expr} denotes reconstructed gene expression, \texttt{Time\_c} is the centered time variable, and \texttt{traj\_id} specifies the each trajectory. This formulation allows us to test for both group-specific baseline expression and time-dependent effects.

We  quantify  the similarity between inferred trajectories using dynamic time warping and identify  two primary trajectory clusters, as illustrated in Figure~\ref{fig:trajectories}c for the Embryoid dataset. The linear mixed-effects modeling identifies genes with significant expression differences between two trajectory clusters (selected genes are shown  in Fig.\ref{fig:trajectories}d). Notably, genes COL1A1, DCN, and LUM are enriched in the first trajectory cluster and remain stably low in the second, all three are canonical markers of mesenchymal/mesodermal derivatives \citep{chen2021type, zoeller2009central, farkas2023derivation}. In contrast, trajectory 2 preferentially expresses genes PAX6, IGFBP2, and TUBA1A, a combination characteristic of neuroectodermal or neural-crest differentiation \citep{zhang2010pax6, shahin2023redox, aiken2017alpha}. These signature genes support the interpretation of  trajectory type I as a mesodermal branch and trajectory type II as a neuroectodermal branch.

 To identify early signature marker genes that are informative for later differentiation fates, we detect genes enriched in the starting cells of each trajectory type using Scanpy’s built-in function \texttt{sc.tl.rank\_genes\_groups}.in~\ref{sec:trajanalysis} (Supplementary Table~\ref{tab:signature_genes}). Specifically, KRT18, KRT19, CD9, and TPM1 are enriched in early state of the mesodermal branch, in agreement with previous reports \citep{corcoran1997keratin, stavish2020generation, wilken2024tropomyosin}. In contrast, genes like TUBA1A, VIM, and FRZB emerged as early markers of the neuroectoderm/neural-crest branch, consistent with earlier findings \citep{chen2018increased, leyns1997frzb}. 

Applying similar analysis to the Statefate dataset yields two dominant trajectory clusters (Supplementary Figure~\ref{fig:trajectories2}). Early signature genes for each branch are provided in Supplementary Table~\ref{tab:signature_genes2}. The genes identified for the trajectory group I point very clearly to differentiation down the granulocytic (neutrophil-like) branch of the myeloid lineage and the acquisition of mature neutrophil–granule and innate-immune functions. Among these genes, 
the coordinated up-regulation of proteases (Prtn3, Elane, Ctsg), peroxidase (Mpo), regulators of protease activity (Slpi), and oxidative-stress enzymes (Gstm1) is a standard signature of neutrophil (granulocyte) maturation \citep{borregaard1997neutrophil}.  Sirpa and Plac8 mark the emerging myeloid identity and terminal differentiation, while metabolic enzymes like Gda and factors such as Dmkn reflect the simultaneous metabolic and microenvironmental remodeling that occurs as HSPCs commit to, and mature within, the granulocytic lineage \citep{zhang2008plac8}. In contrast, genes associated with trajectory II 
marks a much earlier, multipotent progenitor state with simultaneous lineage‐priming rather than terminal specialization \citep{paul2015transcriptional}. 

We apply the same trajectory-based analysis to the Reprogramming dataset, which profiles  mouse embryonic fibroblasts undergoing OKSM-induced reprogramming. For this data set, Doxycycline (Dox) is added at Day 0 and withdrawn at Day 8, after which cells are  cultured in either serum or 2i conditions and profiled by scRNA-seq through Day 18. As shown in Supplementary Figure~\ref{fig:trajectories3}b, developmental trajectories remain similar during the Dox phase (Day 0–8). Focusing on the Day 8–18 serum conditions, we identify three dominant trajectory clusters (Supplementary Figure~\ref{fig:trajectories3}c). The first trajectory type, enriched for Id3 and Col1a1, marks the emergence of stromal-like cells \citep{chen2021type}. The second trajectory type shows higher expression of Cldn6, Sox11, and Krt8, consistent with a combination of neural-like, trophoblast-like, and epithelial-like fates \citep{sugimoto2013tight,jankowski2006sry,golob2019high}. The third trajectory shares features of the first two and likely corresponds to intermediate or transitional states. These patterns recapitulate the key findings reported in the original study \citep{schiebinger2019optimal}.


\subsection{Sensitivity Analysis}
\label{sec:sensitivity}

To learn the robust local Fr\'echet regression model, we tune two key hyperparameters that control the bias-variance trade-off: the bandwidth $h$ and the unbalanced tolerance $\tau$. Generally, smaller $h$ and larger $\tau$ improve estimation accuracy but increase variability. We conduct a sensitivity analysis to guide the choice of these two parameters.

Table~\ref{sen:reprog10} reports the discrepancies between the predicted and ground truth distributions under different bandwidth choices $h$ for two time points for the Reprogramming dataset. At $t=10$, while the performance still exceeds the baseline method, increasing the bandwidth leads to a decrease in accuracy. This is likely due to strong temporal correlation around $t=10$, where putting more weight to distant distributions affects the prediction accuracy. A similar trend is observed in the Embryoid dataset (Table~\ref{sup:sen:EB7}, Table~\ref{sup:sen:EB13}), where smaller bandwidth achieves better predictive performance, though the difference is less pronounced. However, an opposite pattern appears at $t=16.5$, where a larger bandwidth yields a better prediction. This can be partially explained by Figure~\ref{fig:supp2}, which shows that adjacent time points fail to predict the held-out distribution faithfully. In such cases, leveraging information from more distant distributions improves estimation.

Overall, we recommend examining the temporal pattern of distributions within the datasets when selecting $h$. A smaller bandwidth is preferable when stronger temporal correlation exists, while a larger bandwidth helps improve robustness when local information is less reliable. Although the bandwidth can affect predictions, the results remain relative stable across different $h$ and outperforms the baseline methods in most cases. Empirically, we find that choosing the bandwidth to assign most weight to the nearest 5 to 10 distributions strikes a good balance between accuracy and stability.

The selection of $\tau$ follows a similar strategy. A smaller $\tau$, allowing a higher degree of unbalancedness, is preferable when the data is noisy and a more conservative estimate is desired. Conversely, a larger $\tau$ is preferred when the distributions are well-aligned and reliable. This is reflected in Table~\ref{sen:reprog10}, where two time points $t=7.5$ and $t=13.5$ exhibit opposite preferences for $\tau$. Specifically, Table~\ref{benchmark:embryoid} shows that for $t=7.5$, the robust UOT formulation performs better, and decreasing $\tau$ improves the prediction. A similar argument explains the reverse behavior at $t=13.5$. 

Overall, in temporal datasets where the data is less noisy, the impact of $\tau$ is less pronounced. The variation in prediction accuracy should remain within the range defined by the difference between the OT and UOT formulations. For the  Reprogramming dataset (Table~\ref{supp:tauReprog}), the differences across different $tau$ are negligible and likely caused by random variation during the training.

\begin{table}[h!]
    \caption{Sensitivity analysis to bandwith $h$ and unbalanced tolerance $\tau$ using (a)  Reprogramming dataset ($t=10$) and sensitivity to bandwidth $h$; (b) Reprogramming dataset ($t=16.5$) and sensitivity to bandwidth $h$; and (c) embryoid dataset and sensitivity to $\tau$. 
    Metrics are averaged over 100 repetitions.
    }
    \vspace{5pt}
    \label{sen:reprog10}
    \centering
    \begin{tabular}{l|ccc|ccc|ccc}
        \toprule
       \multicolumn{1}{c}{} &\multicolumn{9}{c}{
       Reprogramming dataset ($t=10$)}\\
        & \multicolumn{3}{c|}{$h=0.25$} & \multicolumn{3}{c|}{$h=0.5$} & \multicolumn{3}{c}{$h=1.0$} \\
        \cmidrule(lr){2-4} \cmidrule(lr){5-7} \cmidrule(lr){8-10}
        &  MMD(G) & EMD & $W_2$ & MMD(G) & EMD & $W_2$ & MMD(G) & EMD & $W_2$ \\
        \midrule
        OT          & 0.086 & 4.379 & 4.539 & 0.092 & 4.446 & 4.621 & 0.151 & 4.772 & 4.909 \\
        UOT         & \textbf{0.060} & \textbf{4.331} & \textbf{4.523} & \textbf{0.079} & \textbf{4.389} & \textbf{4.559} & \textbf{0.146} & \textbf{4.701} & \textbf{4.871} \\
        Baseline 1    & 0.128 & 5.059 & 5.307 & 0.128 & 5.059 & 5.307& 0.128 & 5.059 & 5.307 \\
        \toprule
        \multicolumn{1}{c}{} &\multicolumn{9}{c}{
       Reprogramming dataset ($t=16.5$)}\\
        & \multicolumn{3}{c|}{$h=0.25$} & \multicolumn{3}{c|}{$h=0.5$} & \multicolumn{3}{c}{$h=1.0$} \\
        \cmidrule(lr){2-4} \cmidrule(lr){5-7} \cmidrule(lr){8-10}
        &  MMD(G) & EMD & $W_2$ & MMD(G) & EMD & $W_2$ & MMD(G) & EMD & $W_2$ \\
        \midrule
        OT          & 0.052 & 4.108 & 4.700 & 0.051 & 4.017 & 4.552 & \textbf{0.029} & 3.836 & 4.308 \\
        UOT         & \textbf{0.043} & \textbf{4.021} & \textbf{4.624} & \textbf{0.044} & \textbf{3.948} & \textbf{4.481} & 0.030 & \textbf{3.832} & \textbf{4.298} \\
        Baseline 1    & 0.198 & 5.203 & 5.814 & 0.198 & 5.203 & 5.814 & 0.198 & 5.203 & 5.814 \\
        \toprule
        \multicolumn{1}{c}{} &\multicolumn{9}{c}{
       Embryoid data set}\\
        & \multicolumn{3}{c|}{$\tau=1$} & \multicolumn{3}{c|}{$\tau=5$} & \multicolumn{3}{c}{$\tau=10$} \\
        \cmidrule(lr){2-4} \cmidrule(lr){5-7} \cmidrule(lr){8-10}
        &  MMD(G) & EMD & $W_2$ & MMD(G) & EMD & $W_2$ & MMD(G) & EMD & $W_2$ \\
        \midrule
        UOT($t=7.5$)  & 0.248 & 4.787 & 4.955 & 0.281 & 5.046 & 5.223 & 0.278 & 5.175 & 5.357 \\
        UOT($t=13.5$)         & 0.135 & 5.288 & 5.737 & 0.111 & 5.121 & 5.424 & 0.113 & 5.082 & 5.373 \\
        \bottomrule
    \end{tabular}
\end{table}

\section{Discussion}\label{Discussion}
Motivated by the analysis of dynamic single-cell gene expression data, we have developed computationally efficient and robust methods for Fréchet regression. We formulated this as a weighted barycenter estimation problem over high-dimensional probability measures. At the core of our approach is a reparametrization of the barycenter via a generative model, which converts an intractable optimization over probability measures into a more manageable optimization over the model’s parameters. For example, when learning a barycenter in single cell expression space, we parameterize a neural network to generate cells rather than directly parameterizing each constituent cell population.

We have evaluated the proposed methods  on three scRNA-seq datasets, each reflecting a distinct dynamic differentiation process. In all cases, our method produced an effective generative model capable of extrapolating cellular distributions from limited observations, and the resulting transport maps enabled us to track individual cell trajectories throughout the biological process. 

Although our focus here is local Fréchet regression with time as the sole covariate, the same framework readily extends to multivariable covariates by incorporating either a local covariate kernel or the parametric kernels introduced in \citet{petersen2019frechet}. For such population level scRNA-seq data, the Frechet regression provides an effective model to understand how cell distribution changes as a function of the covariates. The neural network-based reprameterization allows us to calculate the weighted barycenter of hundreds of high dimensional distributions. Due to noisy nature of single cell data and mixture of many cell types,  we expect that the use of the unbalance OT in the Fr\'echet regression and pre-training with VAE-NF will lead to more stable estimates of the cell distributions in the Fr\'echet regression. We plant to further investigate such applications. 

\section*{Acknowledgments} This research is supported by NIH grants GM129781 and P30DK050306.

\bigskip
\begin{center}
{\bf SUPPLEMENTARY MATERIAL}
\end{center}
 Supplementary  Material includes  proof of theorem 1, details of 
simulations and  results, and additional plots and figures related to the applications.

\bibliographystyle{imsart-nameyear}
\bibliography{ref}

\newpage

\appendix

\section{Supplementary Methods}

\subsection{Proof of Theorem 1}

\begin{proof}[Proof of Theorem~\ref{thm:fix_point}]
Fix $\mu\!\in\!\mathcal P_{2}(\R^{d})$ and weights $\alpha_1,\dots,\alpha_N\ge 0$ summing to one.
For each $i=1,\dots,N$ choose an optimal \emph{unbalanced} coupling
\[
\gamma_i^\star\in\Pi(\mu,\nu_i),
\qquad
W^{2}_{2,\mathrm{ub}}(\mu,\nu_i)=
\frac12\!\int_{\mathcal{X}\times\mathcal{Y}}\!\|x-y\|^{2}\,d\gamma_i^\star
+\tau\,D_\psi\!\bigl(\gamma_{i,1}^\star\|\nu_i\bigr).
\]

Let $\gamma_{i,x}^\star$ be the conditional law of $y$ given $x$ under $\gamma_i^\star$ and set
\[
T_i(x):=\int_{\mathcal{Y}}y\,d\gamma_{i,x}^\star(y),
\qquad
\bar T(x):=\sum_{i=1}^{N}\alpha_i T_i(x),
\qquad
\bar T(\mu):=\bar T_\#\mu.
\]

For any $x,y\in\R^{d}$ the quadratic identity
\[
\|x-y\|^{2}= \|x-\bar T(x)\|^{2}+\|\bar T(x)-y\|^{2}
            +2\langle\bar T(x)-y,\;x-\bar T(x)\rangle
\]
holds.  Taking $y\sim\gamma_{i,x}^\star$ removes the cross‐term because
$\int(y-T_i(x))\,d\gamma_{i,x}^\star(y)=0$.  Multiplying by $\alpha_i$ and summing over $i$
gives
\[
\sum_{i=1}^{N}\alpha_i\!\int_{\mathcal{Y}}\|x-y\|^{2}\,d\gamma_{i,x}^\star
  =\|x-\bar T(x)\|^{2}
  +\sum_{i=1}^{N}\alpha_i\!\int_{\mathcal{Y}}\|\bar T(x)-y\|^{2}\,d\gamma_{i,x}^\star.
\]
Since the first term on the right is non-negative, discarding it and integrating $x$ w.r.t.~$\mu$
yields
\begin{equation}\label{eq:transport-gap}
\sum_{i=1}^{N}\alpha_i
     \int_{\mathcal{X}\times\mathcal{Y}}\|x-y\|^{2}\,d\gamma_i^\star
\;\ge\;
\sum_{i=1}^{N}\alpha_i
     \int_{\mathcal{X}\times\mathcal{Y}}\|\bar T(x)-y\|^{2}\,d\gamma_i^\star.
\end{equation}

Next observe that pushing the first coordinate of $\gamma_i^\star$ through $\bar T$
produces $\tilde\gamma_i:=(\bar T(x),I)_\#\gamma_i^\star$ which lies in
$\Pi\!\bigl(\bar T(\mu))$.  Hence
\begin{equation}\label{eq:aux-upper}
\frac12\!\int_{\mathcal{X}\times\mathcal{Y}}\!\|\bar T(x)-y\|^{2}\,d\gamma_i^\star
 +\tau\,D_\psi\!\bigl(\gamma_{i,1}^\star\|\nu_i\bigr)
 \ge
 W^{2}_{2,\mathrm{ub}}\!\bigl(\bar T(\mu),\nu_i\bigr).
\end{equation}

Using the definition of $V$ and the optimality of $\gamma_i^\star$,
\[
V(\mu)-V\!\bigl(\bar T(\mu)\bigr)
  =\sum_{i=1}^{N}\alpha_i
     \Bigl[
       \frac12\!\int\|x-y\|^{2}\,d\gamma_i^\star
       +\tau\,D_\psi\!\bigl(\gamma_{i,1}^\star\|\nu_i\bigr)
       -W^{2}_{2,\mathrm{ub}}\!\bigl(\bar T(\mu),\nu_i\bigr)
     \Bigr].
\]
Insert \eqref{eq:aux-upper}, cancel the equal divergence terms, and apply
\eqref{eq:transport-gap}.  The result is
\[
V(\mu)-V\!\bigl(\bar T(\mu)\bigr)
  \ge
  \frac12\int_{\mathcal{X}}\|x-\bar T(x)\|^{2}\,d\mu(x)
  \ge 0.
\]

Finally, equality holds iff the integrand in the last expression vanishes $\mu$-a.s.,
i.e.\ $\bar T(x)=x$ $\mu$-a.s., which implies $\bar T(\mu)=\mu$ and forces equality
throughout.  Conversely, if $\bar T(\mu)=\mu$ the chain of inequalities is tight.
\end{proof}

\subsection{Details of the simulation}
\label{supp:simulation}

To demonstrate that our fixed-point algorithm is robust to outliers and yields faithful representations of the main modes, we evaluate the quality of the computed barycenter on synthetically generated data. Each distribution is modeled as a Gaussian mixture, which reflects the heterogeneous nature of single-cell datasets—where different cell types occupy distinct regions in transcriptomic space. We simulate ten 10-dimensional Gaussian mixtures and compute their barycenter using the proposed fixed-point method.

To generate ten Gaussian mixtures, we first set up a template mixture with pre-specified means, covariances, and component weights. Each observed distribution is then created by perturbing the parameters of this template. The template consists of eight components: the first four represent the main modes, and the remaining four correspond to outliers. The component means are set as \([2, \ldots, 2]\), \([2, \ldots, 2, -2, \ldots, -2]\), \([-2, \ldots, -2, 2, \ldots, 2]\), and \([-2, \ldots, -2]\) for the main modes, and \([10, \ldots, 10, 0, \ldots, 0]\), \([-10, \ldots, -10, 0, \ldots, 0]\), \([0, \ldots, 0, 10, \ldots, 10]\), and \([0, \ldots, 0, -10, \ldots, -10]\) for the outliers. All components share the same covariance matrix \(0.5I\), and the weights are assigned as \\
\([0.24, 0.24, 0.23, 0.23, 0.015, 0.015, 0.015, 0.015]\), with the outlier components comprising only 6\% of the total mass. 

To generate the ten observed Gaussian mixtures, we introduce small Gaussian perturbations to the means, covariances, and weights of the template while ensuring validity—that is, positive definite covariance matrices and non-negative weights summing to one. We sample 3,000 data points from each mixture to serve as input for barycenter computation. The OT barycenter is computed using the fixed-point algorithm proposed by \citet{korotin2022wasserstein}, while the UOT barycenter is obtained using our method with the unbalanced penalty parameter set to $\tau=1$.

\subsection{The COVID dataset}
\label{supp:COVID}

The dataset \citep{melms2021COVID} profiles lung tissues from nineteen individuals who died of COVID-19 and underwent rapid autopsy, along with seven control individuals, to investigate disease-associated transcriptional changes. About 116,000 nuclei were profiled using single-nucleus RNA sequencing (snRNA-seq) on the 10x Genomics Chromium platform. We downloaded the processed and log-transformed dataset from \url{https://portals.broadinstitute.org/single_cell}. We then extracted the top 2000 HVGs, and computed top 50 PCs. When evaluating the effectiveness of the pretraining strategy, we compute the unbalanced barycenter ($\tau=5$) of the 19 individuals who died of COVID-19 in the 50-dimensional PC space.

\section{Supplementary Tables}

\renewcommand{\thetable}{S\arabic{table}}
\setcounter{table}{0}

\begin{table}[h!]
    \caption{Benchmark on Embryoid Dataset ($t=7.5$): Sensitivity to Bandwidth $h$. Metrics are averaged over 100 repetitions.}
    \vspace{5pt}
    \label{sup:sen:EB7}
    \centering
    \begin{tabular}{l|ccc|ccc|ccc}
        \toprule
        & \multicolumn{3}{c|}{$h=3$} & \multicolumn{3}{c|}{$h=7$} & \multicolumn{3}{c}{$h=15$} \\
        \cmidrule(lr){2-4} \cmidrule(lr){5-7} \cmidrule(lr){8-10}
        &  MMD(G) & EMD & $W_2$ & MMD(G) & EMD & $W_2$ & MMD(G) & EMD & $W_2$ \\
        \midrule
        OT          & \textbf{0.280} & 5.260 & 5.464 & 0.307 & 5.326 & 5.531 & \textbf{0.319} & 5.345 & 5.537 \\
        UOT         & 0.281 & \textbf{5.046} & \textbf{5.223} & \textbf{0.305} & \textbf{5.114} & \textbf{5.282} & 0.348 & \textbf{5.127} & \textbf{5.302} \\
        Baseline 1    & 0.381 & 5.370 & 5.558 & 0.381 & 5.370 & 5.558& 0.381 & 5.370 & 5.558 \\
        \bottomrule
    \end{tabular}
\end{table}

\begin{table}[h!]
    \caption{Benchmark on Embryoid Dataset ($t=13.5$): Sensitivity to Bandwidth $h$. Metrics are averaged over 100 repetitions.}
    \vspace{5pt}
    \label{sup:sen:EB13}
    \centering
    \begin{tabular}{l|ccc|ccc|ccc}
        \toprule
        & \multicolumn{3}{c|}{$h=3$} & \multicolumn{3}{c|}{$h=7$} & \multicolumn{3}{c}{$h=15$} \\
        \cmidrule(lr){2-4} \cmidrule(lr){5-7} \cmidrule(lr){8-10}
        &  MMD(G) & EMD & $W_2$ & MMD(G) & EMD & $W_2$ & MMD(G) & EMD & $W_2$ \\
        \midrule
        OT         & 0.114 & \textbf{5.076} & \textbf{5.319} & 0.189 & 5.233 & 5.522 & \textbf{0.206} & \textbf{5.382} & \textbf{5.669} \\
        UOT         & \textbf{0.111} & 5.121 & 5.424 & \textbf{0.166} & \textbf{5.184} & \textbf{5.520} & 0.234 & 5.518 & 5.869 \\
        Baseline 1    & 0.257 & 5.535 & 5.607 & 0.257 & 5.535 & 5.607 & 0.257 & 5.535 & 5.607 \\
        \bottomrule
    \end{tabular}
\end{table}

\begin{table}[h!]
    \caption{Benchmark on Reprogramming Dataset: Sensitivity to Unbalanced Tolerance $\tau$. Metrics are averaged over 100 repetitions.}
    \vspace{5pt}
    \label{supp:tauReprog}
    \centering
    \begin{tabular}{l|ccc|ccc|ccc}
        \toprule
        & \multicolumn{3}{c|}{$\tau=1$} & \multicolumn{3}{c|}{$\tau=5$} & \multicolumn{3}{c}{$\tau=10$} \\
        \cmidrule(lr){2-4} \cmidrule(lr){5-7} \cmidrule(lr){8-10}
        &  MMD(G) & EMD & $W_2$ & MMD(G) & EMD & $W_2$ & MMD(G) & EMD & $W_2$ \\
        \midrule
        UOT($t=10$)  & 0.043 & 3.920 & 4.415 & 0.044 & 3.948 & 4.481 & 0.049 & 3.927 & 4.426 \\
        UOT($t=16.5$)         & 0.068 & 4.267 & 4.455 & 0.079 & 4.446 & 4.621 & 0.077 & 4.319 & 4.510 \\
        \bottomrule
    \end{tabular}
\end{table}

\begin{table}[H]
\centering
\caption{Top 10 signature genes enriched in the starting cells of each trajectory group in the Embryoid dataset. Mean log-expression is averaged across starting cells in each group. Score is the Z-statistic from the Wilcoxon rank-sum test.}
\label{tab:signature_genes}
\begin{tabular}{lrr|lrr}
\toprule
\multicolumn{3}{c|}{Group I} & \multicolumn{3}{c}{Group II} \\
Gene & Mean Expr. & Score & Gene & Mean Expr. & Score \\
\midrule
EPCAM   & 1.30 & 22.33  & TUBA1B  & 3.94 & 20.74 \\
PDLIM1  & 0.98 & 20.54  & VIM     & 1.71 & 20.64 \\
KRT18   & 2.50 & 19.01  & TPBG    & 1.15 & 19.31 \\
S100A10 & 0.39 & 18.88  & TUBA1A  & 2.27 & 17.52 \\
TDGF1   & 1.34 & 18.67  & PTN     & 1.11 & 17.25 \\
TPM1    & 1.69 & 18.41  & CRABP2  & 1.85 & 16.93 \\
CD9     & 0.97 & 17.77  & HMGB2   & 2.20 & 16.61 \\
MGST1   & 0.91 & 17.76  & GPC3    & 2.02 & 16.29 \\
POU5F1  & 1.66 & 17.49  & FRZB    & 0.70 & 16.20 \\
KRT19   & 1.31 & 17.27  & DLK1    & 1.20 & 16.17 \\
\bottomrule
\end{tabular}
\end{table}

\begin{table}[H]
\centering
\caption{Top 10 signature genes enriched in the starting cells of each trajectory group in the Statefate dataset. Mean log-expression is averaged across starting cells in each group. Score is the Z-statistic from the Wilcoxon rank-sum test.}
\label{tab:signature_genes2}
\begin{tabular}{lrr|lrr}
\toprule
\multicolumn{3}{c|}{Group I} & \multicolumn{3}{c}{Group II} \\
Gene & Mean Expr. & Score & Gene & Mean Expr. & Score \\
\midrule
Prtn3      & 2.62 & 36.42 & Gata2      & 1.72 & 29.54 \\
Elane      & 2.93 & 33.70 & F2r        & 1.35 & 26.10 \\
Plac8      & 3.35 & 31.06 & Alox5      & 1.45 & 25.84 \\
Mpo        & 2.13 & 28.91 & Serpina3g  & 1.74 & 22.07 \\
Ctsg       & 2.48 & 28.30 & Scin       & 1.12 & 21.84 \\
Dmkn       & 1.40 & 27.98 & Csf2rb2    & 1.53 & 21.29 \\
Gda        & 1.31 & 27.42 & Ctla2a     & 0.97 & 20.50 \\
Gstml      & 1.76 & 25.84 & Ptprcap    & 1.32 & 20.32 \\
Slpi       & 1.23 & 25.06 & Ikzf2      & 0.95 & 20.14 \\
Sirpa      & 1.29 & 24.75 & Car2       & 1.16 & 19.65 \\
\bottomrule
\end{tabular}
\end{table}


\section{Additional Figures}

\renewcommand{\thefigure}{S\arabic{figure}}
\setcounter{figure}{0}

\begin{figure}[H]
    \centering
    \includegraphics[width=0.95\linewidth]{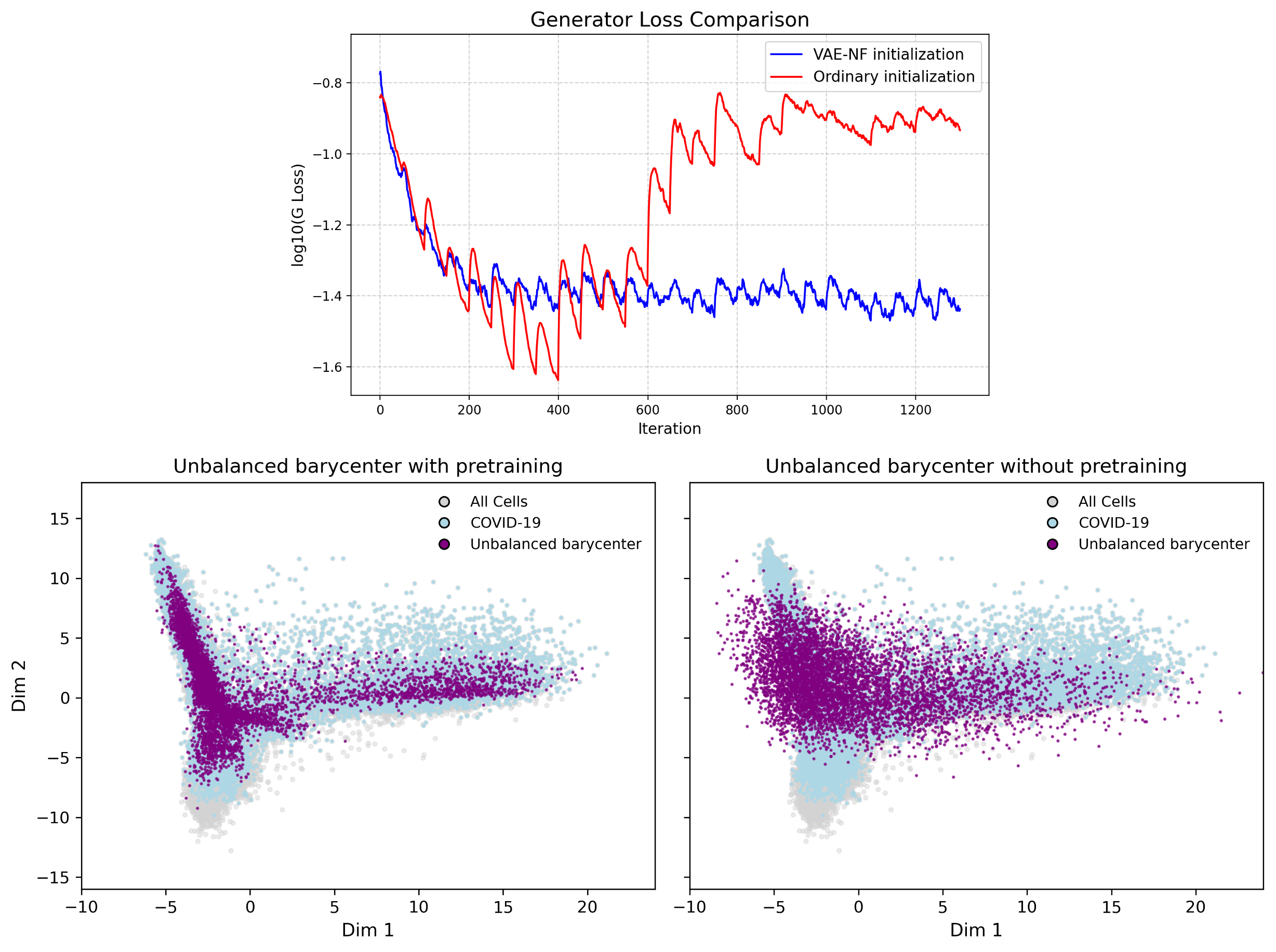}
    \caption{
    \textbf{Top:} Generator loss during training, comparing two initialization strategies. The y-axis shows the log-transformed regression loss. Lower and more stable loss indicates better convergence of the generator to the fixed point. \textbf{Bottom:} Learned unbalanced barycenters for COVID-19 subjects in the COVID dataset (see Supplement~\ref{supp:COVID}), under two different initialization schemes. Grey points represent pooled cells from all conditions, light blue points denote COVID-19 samples, and purple points show the learned barycenter. Pretraining with VAE-NF leads to a more structured and concentrated barycenter compared to standard initialization.
    }
    \label{fig:supp_pretrain}
\end{figure}

\begin{figure}[H]
    \centering
    \includegraphics[width=1\linewidth]{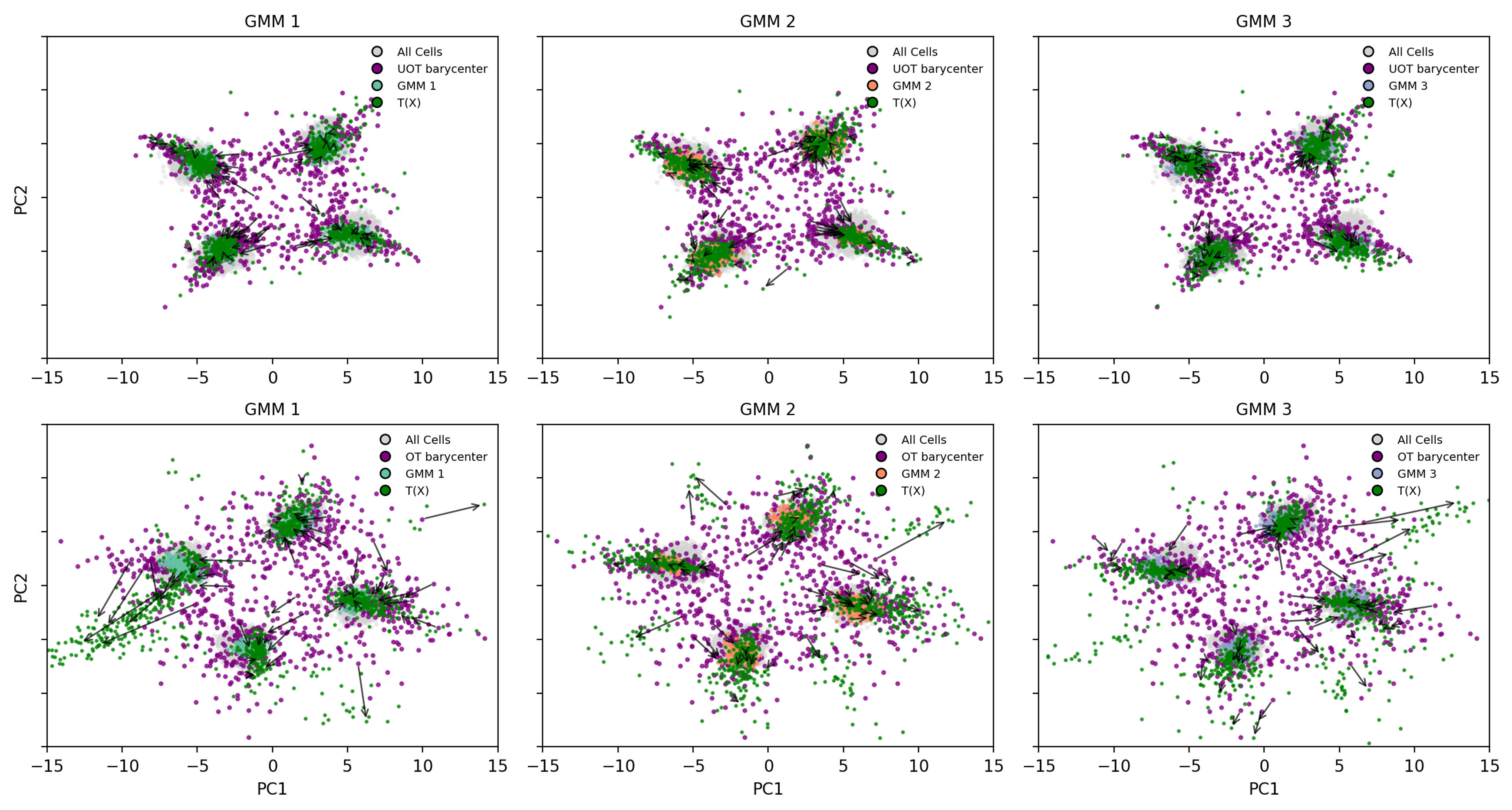}
    \caption{
    Transport plans from the barycenter to three selected Gaussian mixtures. Arrows indicate the learned mappings from barycenter samples to the target mixture. Grey points represent pooled samples from all ten mixtures; purple points denote the learned barycenter; green points are pushforward samples; and colored points labeled “GMM $i$” show the ground truth distribution of each mixture.
    \textbf{Top row:} unbalanced optimal transport (UOT) from the robust Fr\'echet barycenter. 
    \textbf{Bottom row:} classical balanced optimal transport (OT) from the Wasserstein barycenter.
    }
    \label{fig:supp_simu}
\end{figure}

\begin{figure}[H]
    \centering
    \includegraphics[width=0.9\linewidth]{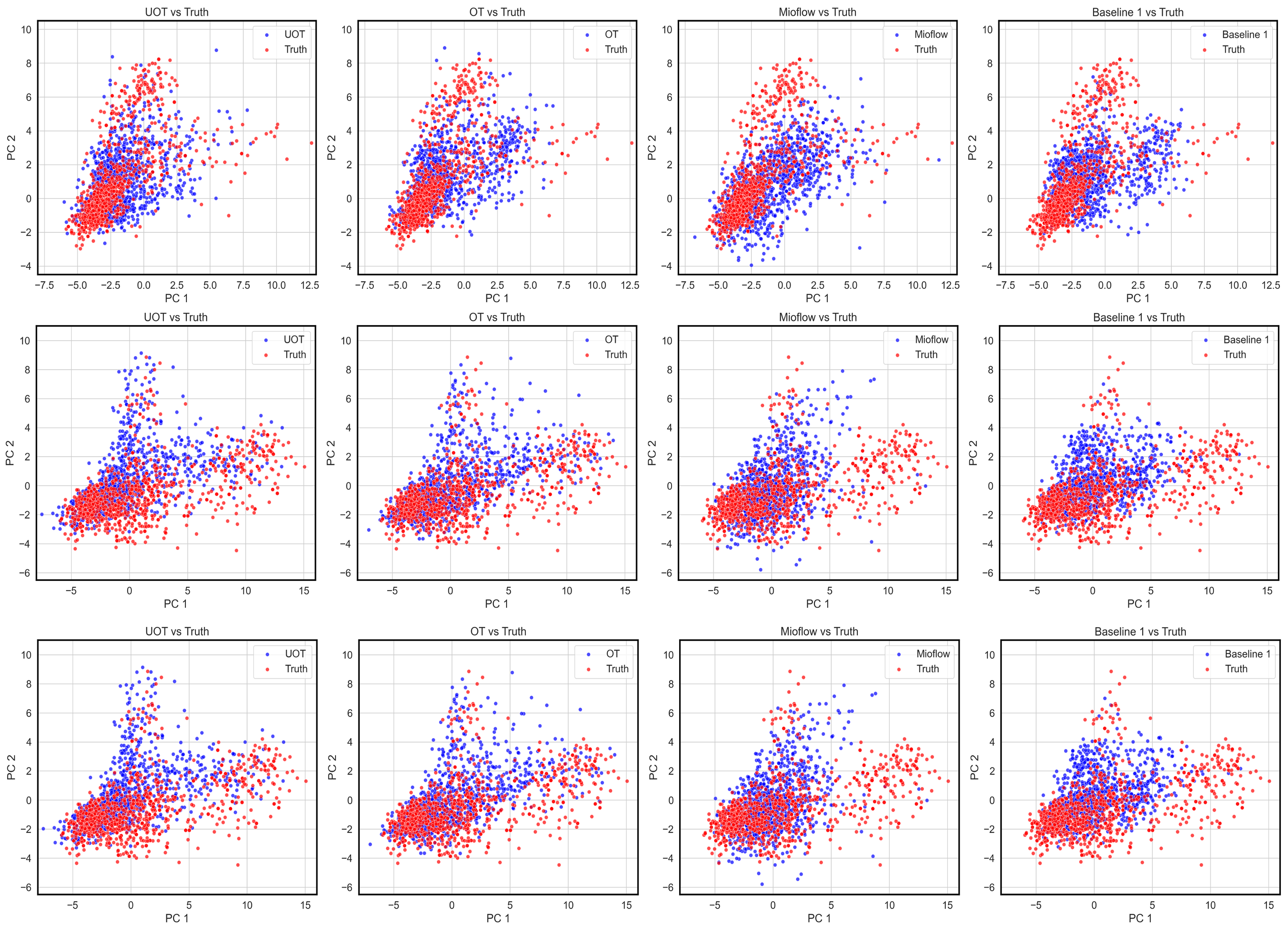}
    \caption{Comparison of the predicted cellular distributions (blue) against the ground truth (red) visualized on the top two principal components (PC1 and PC2) across different methods for the Embryoid dataset. \textbf{First row:} \(t = 7.5\); \textbf{Second row:} \(t = 13.5\); \textbf{Third row:} \(t = 19.5\).}
    \label{fig:supp1}
\end{figure}

\begin{figure}[H]
    \centering
    \includegraphics[width=0.9\linewidth]{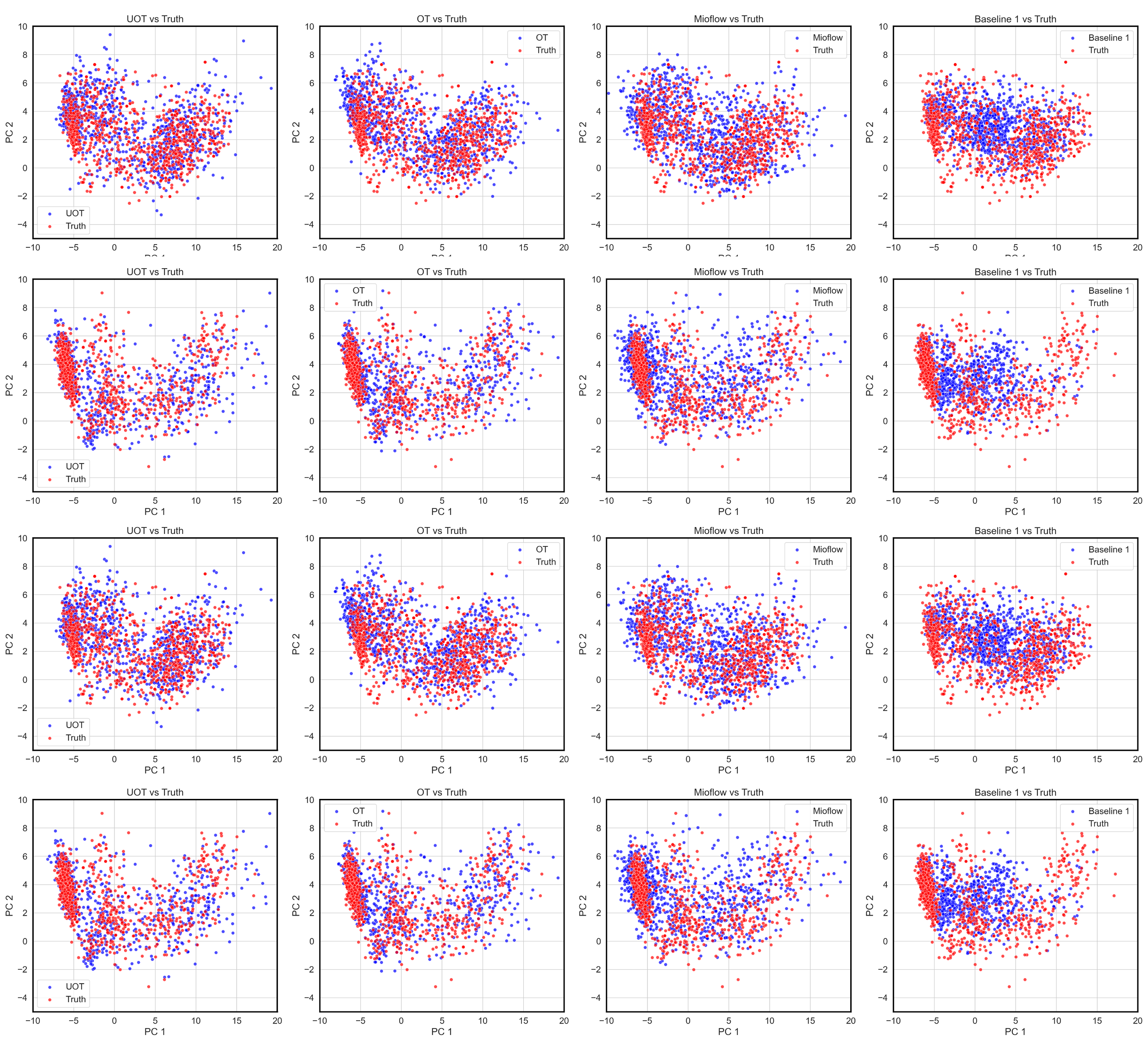}
    \caption{Comparison of the predicted cellular distributions (blue) against the ground truth (red) visualized on the top two principal components (PC1 and PC2) across different methods for the Reprogramming dataset. \textbf{First row:} \(t = 6\); \textbf{Second row:} \(t = 10\); \textbf{Third row:} \(t = 13.5\); \textbf{Fourth row:} \(t = 16.5\)}
    \label{fig:supp2}
\end{figure}

\begin{figure}[H]
    \centering
    \includegraphics[width=\linewidth]{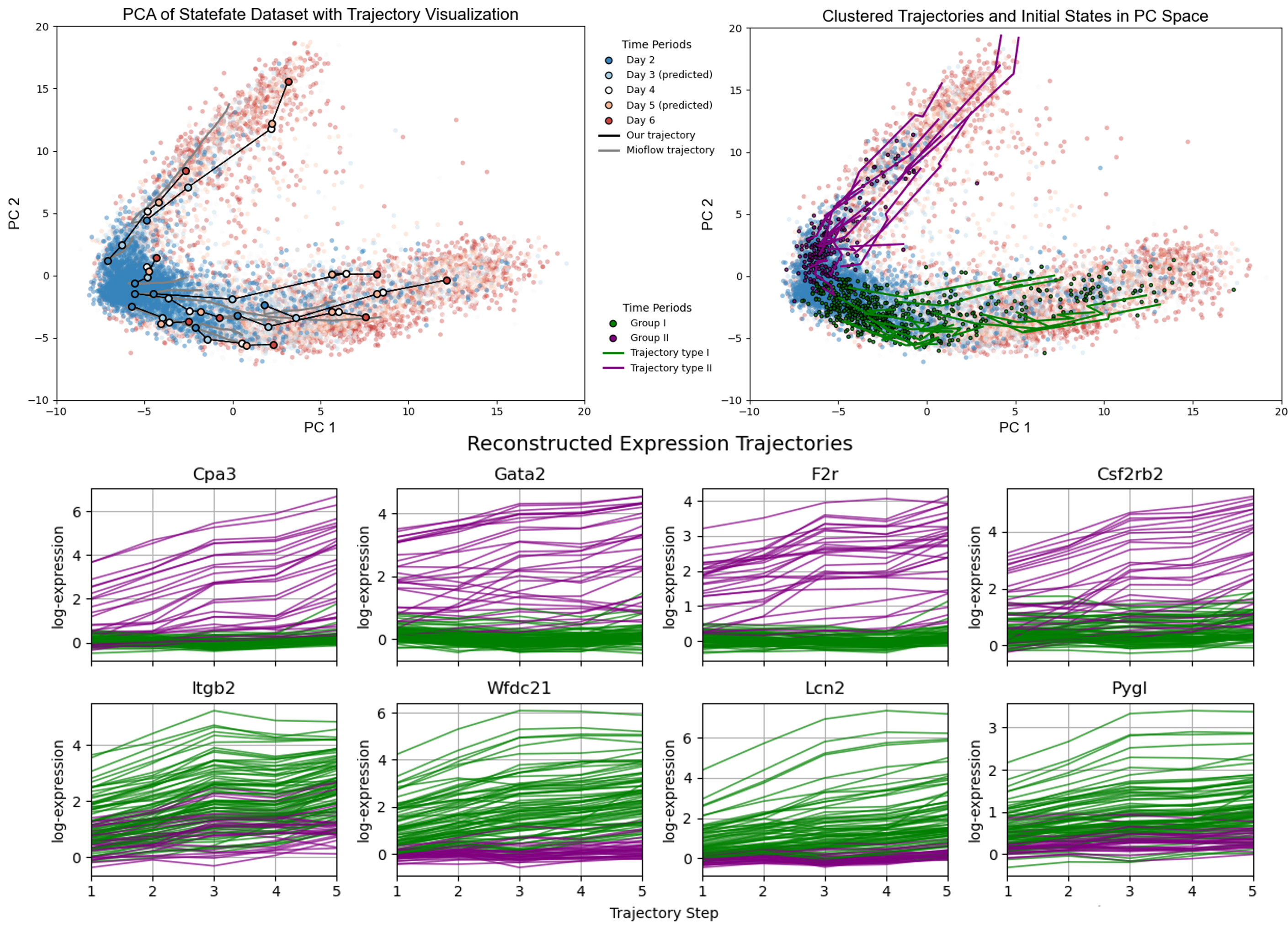}
    \caption{Distributional interpolation and trajectory reconstruction in the Statefate dataset, visualized using the top two principal components (PC1 and PC2). \textbf{Top left,} Reconstructed trajectories from Day 2 to Day 6. Black lines trace our predicted trajectories; Each dot on them represents the predicted position of a cell in transcriptomic space at a given time point. The grey lines correspond to Mioflow predictions for the same cells. \textbf{Top right,} Two differentiation fates identified via Dynamic Time Warping (DTW) clustering. Initial states are colored by group (Group I: green; Group II: purple). \textbf{Bottom,} Reconstructed gene expression trajectories for genes showing significant differences in intercept or slope across trajectory types.}
    \label{fig:trajectories2}
\end{figure}

\begin{figure}[H]
    \centering
    \includegraphics[width=\linewidth]{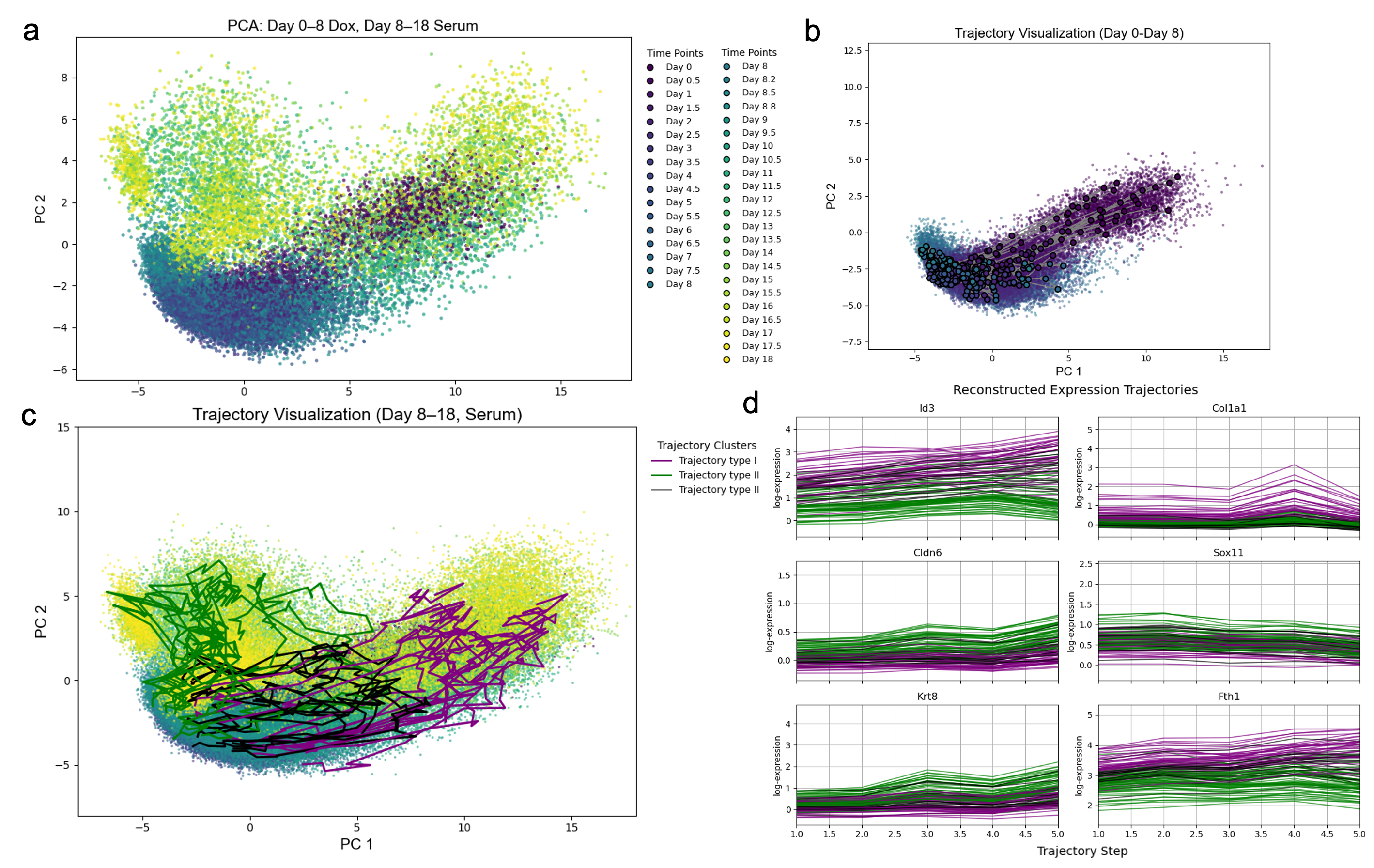}
    \caption{Trajectory analysis of cellular reprogramming dynamics. 
    \textbf{a,} Single-cell profiles from Day 0 to Day 8 (Dox condition) and Day 8 to Day 18 (serum condition), visualized using the top two PCs. 
    \textbf{b,} Developmental trajectories inferred for the early phase (Day 0–Day 8). 
    \textbf{c,} Three distinct differentiation trajectories identified from Day 8 to Day 18 under serum conditions using Dynamic Time Warping (DTW) clustering. 
    \textbf{d,} Reconstructed gene expression trajectories for representative marker genes, illustrating dynamic differences in expression patterns (intercepts and slopes) across trajectory types.}
    \label{fig:trajectories3}
\end{figure}

\end{document}